\RequirePackage{luatex85}

\documentclass[11pt,a4paper]{article}
\usepackage{iftex}

\ifPDFTeX
\usepackage[utf8]{inputenc}
\usepackage[T1]{fontenc}
\fi

\usepackage{amsmath}
\usepackage{amssymb}
\usepackage{amsthm}
\usepackage{mathtools}
\usepackage{thmtools}

\ifPDFTeX
\usepackage{libertine}
\usepackage[libertine]{newtxmath} 
\usepackage[scaled=0.96]{zi4} 
\else
\usepackage{fontspec}
\usepackage{unicode-math}
\fi

\usepackage[DIV=10]{typearea} 

\usepackage{microtype}

\usepackage[procnumbered,ruled,vlined,linesnumbered,algo2e]{algorithm2e}

\usepackage{xcolor}
\usepackage{tikz}

\usepackage{comment}

\usepackage[ocgcolorlinks]{hyperref} 

\colorlet{DarkRed}{red!50!black}
\colorlet{DarkGreen}{green!50!black}
\colorlet{DarkBlue}{blue!50!black}

\hypersetup{
	linkcolor = DarkRed,
	citecolor = DarkGreen,
	urlcolor = DarkBlue,
	bookmarksnumbered = true,
	linktocpage = true
}

\allowdisplaybreaks[1]

\DontPrintSemicolon
\SetKwProg{Procedure}{Procedure}{}{}
\SetProcNameSty{textsc}
\SetFuncSty{textsc}
\SetKw{KwAnd}{and}
\SetKw{KwDo}{do}

\SetCommentSty{mycommfont}

\newtheorem{theorem}{Theorem}[section]
\newtheorem{lemma}[theorem]{Lemma}

\newtheorem{definition}[theorem]{Definition}

\newcommand{\davg}{\overline{d}}
\newcommand{\vol}{\operatorname{vol}}

\title{A Faster Local Algorithm for Detecting Bounded-Size Cuts with Applications to Higher-Connectivity Problems}

\author{Sebastian Forster~\thanks{Department of Computer Sciences, University of Salzburg, Austria. Work partially done while at Max Planck Institute for Informatics and while at University of Vienna.} \and Liu Yang\thanks{Work partially done while at Yale University.}}

\date{}

\hypersetup{
	pdftitle = {A Faster Local Computation Algorithm for Detecting Bounded-Size Cuts with Applications to Higher-Connectivity Problems},
	pdfauthor = {Sebastian Forster, Liu Yang}
}

\begin{document}
\maketitle
\begin{abstract}
Consider the following ``local'' cut-detection problem in a directed graph:
We are given a starting vertex $ s $ and need to detect whether there is a cut with at most $ k $ edges crossing the cut such that the side of the cut containing $ s $ has at most $ \Delta $ interior edges.
If we are given query access to the input graph, then this problem can in principle be solved in sublinear time without reading the whole graph and with query complexity depending on $ k $ and $ \Delta $.
We design an elegant randomized procedure that solves a slack variant of this problem with $ O (k^2 \Delta) $ queries, improving in particular a previous bound of $ O ((2 (k + 1))^{k+2} \Delta) $ by Chechik et al.~[SODA 2017].
In this slack variant, the procedure must successfully detect a component containing~$ s $ with at most $ k $ outgoing edges and $ \Delta $ interior edges if such a component exists, but the component it actually detects may have up to $ O (k \Delta) $ interior edges.

Besides being of interest on its own, such cut-detection procedures have been used in many algorithmic applications for higher-connectivity problems.
Our new cut-detection procedure therefore almost readily implies (1) a faster vertex connectivity algorithm which in particular has nearly linear running time for polylogarithmic value of the vertex connectivity, (2) a faster algorithm for computing the maximal $k$-edge connected subgraphs, and (3) faster property testing algorithms for higher edge and vertex connectivity, which resolves two open problems, one by Goldreich and Ron~[STOC '97] and one by Orenstein and Ron~[TCS 2011].
\end{abstract}

\section{Introduction}

Although the term ``big data'' has been overused in recent years, there is a factual trend in growing input sizes for computational tasks.
In practical applications, this has led to many engineering and architectural efforts for designing scalable systems~\cite{DeanG04,White12,IsardBYBF07,ZahariaCFSS10}.
In theory of computing, this trend has mainly been reflected by studying models that penalize reading, storing, or exchanging data such as the property testing model~\cite{GoldreichGR98}, the streaming model~\cite{HenzingerRR98,McGregor14}, the CONGEST model~\cite{Peleg00}, the congested clique model~\cite{LotkerPPP05}, the massively parallel computation model~\cite{KarloffSV10}, or the $k$-machine model~\cite{KlauckNP015}. 
A recent take on this are \emph{Local Computation Algorithms}~\cite{RubinfeldTVX11,LeviM17}, where the goal is to design sublinear-time algorithms that do not even need to read the whole input.

In this paper, we study a natural ``local'' version of the problem of finding an edge cut of bounded size with applications to higher connectivity: given a starting vertex~$ s $, detect a ``small'' component containing $ s $ with at most $ k $ outgoing edges (if it exists).
This problem has implicitly been studied before in the context of property testing~\cite{GoldreichR02,OrensteinR11} and for developing a centralized algorithm that computes the maximal $k$-edge connected subgraphs of a directed graph~\cite{ChechikHILP17}.
Recently, a variant for vertex cuts has been studied for obtaining faster vertex connectivity algorithms~\cite{NanongkaiSY19}.
A somewhat similar problem of locally detecting a small component with low-conductance has recently been studied extensively~\cite{KawarabayashiT19,HenzingerRW17,SaranurakW19}, in particular to obtain centralized algorithms for deterministically computing the edge connectivity of an undirected graph.

We improve upon the query complexity and running time of all prior local computation algorithms for detecting bounded-size cuts.
The significance of this contribution is confirmed by the fact that it almost readily implies faster algorithms for several problems in higher connectivity.
First, we obtain the first nearly-linear time algorithm for computing the vertex connectivity $ \kappa $ of a directed or undirected graph whenever $ \kappa $ is polylogarithmic in the number of vertices.
Our algorithm is the fastest known algorithm for a wide, polynomial, range of $ \kappa $.
Second, we obtain algorithms for computing the maximal $k$-edge connected subgraphs of directed or undirected graphs that significantly improve the algorithms of Chechik et al.~\cite{ChechikHILP17}.
Third, we improve upon the state of the art in property testing algorithms for higher connectivity in essentially all settings, considering both $k$-edge connectivity and $k$-vertex connectivity, both directed and undirected graphs, and both graphs of bounded and unbounded degree.
Furthermore, we give a much more uniform treatment of all these settings.
In particular, our improvements answer two long-standing open problems in the field posed by Goldreich and Ron~\cite{GoldreichR02} and by Orenstein and Ron~\cite{OrensteinR11}, respectively.

We next settle our notation and terminology in Section~\ref{sec:preliminaries} before we discuss our results and compare them with related work in Section~\ref{sec:discussion of results}.
We then proceed by giving the details of our results: the local cut detection result in Section~\ref{sec:local cut detection}, the vertex connectivity result in Section~\ref{sec:vertex connectivity}, the result on computing the maximal $k$-edge connected subgraphs in Section~\ref{sec:connected subgraphs}, and the property testing results in Section~\ref{sec:property testing}.
Finally, we conclude the paper in Section~\ref{sec:conclusion} with some open problems.

\subsection{Preliminaries}\label{sec:preliminaries}

In the following we fix our terminology and give all relevant definitions.

\paragraph{Graph Terminology}

In this paper, we consider input graphs $ G = (V, E) $ with $ n $ vertices $ V = \{1, \ldots, n \} $ and $ m $ edges $ E $.
Graphs might be directed, i.e., $ E \subseteq V^2 $, or undirected, i.e., $ E \subseteq \binom{V}{2} $.
For any two subsets of vertices $ A \subseteq V $ and $ B \subseteq V $ we define the set of edges going from $ A $ to $ B $ as $ E (A, B) = E \cap (A \times B) $.
Note that $ E (A, B) = E (B, A) $ in undirected graphs.
In directed graphs, we say that an edge $ (u, v) $ is reachable from a vertex $ s $ if $ u $ is reachable from~$ s $.
The reverse of an edge $ (u, v) $ is the edge $ (v, u) $ and the reverse of a graph $ G = (V, E) $ is the graph in which every edge is reversed, i.e., $ G' = (V, E') $ with edge set $ E' = \{ (v, u) : (u, v) \in E \} $.
We say that an event occurs with high probability if it does so with probability at least $ 1 - 1/n^c $ for any (fixed) choice of the constant $ c \geq 1 $.

A graph is \emph{strongly connected} if for every pair of vertices $ u $ and $ v $ there is a path from $ u $ to $ v $ and a path from $ v $ to $ u $. (The former implies the latter in undirected graphs and we usually just call the graph \emph{connected} in that case.)
A graph is \emph{$k$-edge connected} if it is strongly connected whenever fewer than $ k $ edges are removed.
A graph is \emph{$k$-vertex connected} if it has more than $ k $ vertices and is strongly connected whenever fewer than $ k $ vertices (and their incident edges) are removed.
The edge connectivity~$ \lambda $ of a graph is the minimum value of $ k $ such that the graph is $k$-edge connected and the vertex connectivity~$ \kappa $ is the minimum value of $ k $ such that the graph is $k$-vertex connected.
An edge cut $ (L, R) $ is a partition of the vertices into two non-empty sets $ L $ and $ R $.
The size of an edge cut $ (L, R) $ is the number of edges going from $ L $ to $ R $, i.e., $ | E (L, R) | $.
A vertex cut $ (L, M, R) $ is a partition of the vertices into three sets $ L $, $ M $, and $ R $ such that $ L $ and $ R $ are non-empty and there are no edges from $ L $ to $ R $, i.e., $ E (L, R) = \emptyset $.
The size of a vertex cut $ (L, M, R) $ is $ | M | $, the size of $ M $.
Observe that a graph is $ k $-edge connected if and only if it has no edge cut of size at most $ k - 1 $ and it is $ k $-vertex connected if and only if it has no vertex cut of size at most $ k - 1 $.

In our paper, we are mainly interested in detecting so-called $k$-edge-out and -in components and $k$-vertex-out and -in components, which we define in the following.
For every subset of vertices $ C \subseteq V $, the \emph{vertex size of $ C $} $ C $ is $ | C |$, the \emph{edge size of $ C $} is $ | E (C, C) | $, the \emph{volume of $ C $} is $ \vol (C) = | E (C, V) | $, and the \emph{symmetric volume of $ C $} is $ \vol^* (C) = | E (C, V) \cup E (V, C) | $.

\begin{definition}
A \emph{$k$-edge-out (-in) component} of a directed graph $ G = (V, E) $ is a non-empty subset of vertices $ C \subseteq V $ such that there are at most $ k $ edges leaving (entering)~$ C $, i.e., $ | E \cap (C \times V \setminus C) | \leq k $ ($ | E \cap (V \setminus C \times C) | \leq k $).
$ C $ is \emph{minimal} if for every proper subset $ \hat{C} $ of $ C $ the number of edges leaving (entering) $ \hat{C} $ is more than the number of edges leaving (entering)~$ C $.
$ C $ is \emph{proper} if $ C $ is a proper subset of $ V $, i.e., $ C \subset V $.
\end{definition}

Observe that for every proper, non-empty $k$-edge-out component~$ C $ there exists an edge cut $ (C, V \setminus C) $ of size at most $ k $, certifying an edge connectivity of less than $ k $, and for every proper, non-empty $k$-edge-in component~$ C $ there exists an edge cut $ (V \setminus C, C) $ of size at most $ k $, certifying a vertex connectivity of less than $ k $.

\begin{definition}
A \emph{$k$-vertex-out (-in) component} of a directed graph $ G = (V, E) $ is a non-empty subset of vertices $ C \subseteq V $ such that the number of vertices reached from edges leaving (entering)~$ C $ is at most $ k $, i.e., $ |B| \leq k $ for $ B = \{ v \in V \setminus C \mid \exists (u, v) \in E : u \in C \} $ ($ B = \{ u \in V \setminus C \mid \exists (u, v) \in E : v \in C \} $).
$ C $ is \emph{minimal} if for every proper subset $ \hat{C} $ of $ C $ the number of vertices reached from edges leaving (entering)~$ \hat{C} $ is more than the number of vertices reached from edges leaving (entering)~$ C $.
$ C $ is \emph{proper} if $ C \cup B $ is a proper subset of $ V $, i.e., $ C \cup B \subset V $.
\end{definition}
Observe that for every proper, non-empty $k$-vertex-out component~$ C $ there exists a vertex cut $ (C, B, V \setminus C) $ of size at most $ k $ and for every proper, non-empty $k$-vertex-in component~$ C $ there exists an edge cut $ (V \setminus C, B, C) $ of size at most $ k $.

In this paper, we use $ \tilde O (\cdot) $-notation to suppress polylogarithmic factors.

\paragraph{Query Model}

In the \emph{incidence-lists model}, we assume that there is some order on the outgoing as well as on the incoming edges of every vertex.
In directed graphs, the algorithm can, for any vertex $ v $ and any integer $ i \geq 1 $, in a single query ask for the $i$-th outgoing edge of $ v $ (and either get this edge in return or a special symbol like $ \bot $ if $ v $ has fewer than $ i $ outgoing edges).
Similarly, the algorithm can ask for $i$-th incoming edge of $ v $.
In undirected graphs, the algorithm can simply ask for the $i$-th edge incident on $ v $.
The query complexity of an algorithm is the number of queries it performs.

\paragraph{Property Testing}

Informally, the main goal in property testing is to find out whether a given graph, to which the algorithm has query access, has a certain property by performing a sublinear number of queries, allowing a ``soft'' decision margin.
In this paper, we consider the properties of being $ k $-edge connected or being $ k $-vertex connected.

Formally, a graph property $ \mathcal{P} $ is a subset of graphs closed under isomorphism, usually specified implicitly by a predicate $ P $.
In the \emph{bounded-degree model} a graph is said to be $ \epsilon $-far from having the property in question, if more than $ \epsilon n d $ edge modifications must be made to obtain a graph that has the property.
In directed graphs, $ d $ is a given upper bound on the maximum in- or out-degree of any vertex (i.e., the maximum number of incoming or outgoing edges of any vertex).
In undirected graphs, $ d $ is a given upper bound on the maximum degree of any vertex (i.e., the maximum number of edges incident on any vertex).
Note that $ n d $ is an upper bound on~$ m $ in directed graphs and an upper bound on $ 2 m $ in undirected graphs.
In the \emph{unbounded-degree model}, the number of edge modifications must be more than $ \epsilon n \davg $, where $ \davg := m / n $.
Note that $ n \davg $ is an upper bound on $ m $ in directed graphs and an upper bound on $ 2 m $ in undirected graphs.

A \emph{property testing algorithm} has query access to the input graph and, with probability at least $ \tfrac{2}{3} $, is required to accept every graph that has the property and reject every graph that is $ \epsilon $-far from having the property.
We assume that the property testing algorithm knows $ n $ and $ d $ or $ \davg $, respectively, in advance and that $ \epsilon $ and $ k $ are given as parameters.

\subsection{Discussion of Our Results}\label{sec:discussion of results}

In the following we compare our results with prior work.

\subsubsection{Local Cut Detection}

Consider the following ``local'' problem:
We are given a starting vertex $ s $ and need to detect whether there is a $k$-edge-out component containing~$ s $ of vertex size at most $ \Gamma $ (or of edge size at most $ \Delta $) in $ G $.
The goal is to design sublinear-time algorithms that only have query access to the graph and thus answer this question by locally exploring the neighborhood of $ s $ instead of reading the whole graph.
Several algorithms for this problem have been proposed in the literature.
We review them in the following and compare them to our new result.

In undirected graphs, a local variant of Karger's minimum cut algorithm~\cite{Karger00} gives the following guarantees.
\begin{theorem}[Implicit in \cite{GoldreichR02}]
There is a randomized algorithm that, given integer parameters $ k \geq 1 $ and~$ \Delta \geq 1 $, a starting vertex~$ s $, and query access to an undirected graph~$ G $, with $ O (\Gamma^{2 - 2/k} (\Delta + k)) $ queries detects, with constant success probability, whether there is a $k$-edge-out component containing~$ s $ of vertex size at most~$ \Gamma $ and edge size at most~$ \Delta $ in~$ G $ and if so returns such a component.
\end{theorem}

In directed graphs, a backtracking approach that tries to identify the at most $ k $ edges leaving the component gives the following guarantees.

\begin{theorem}[\cite{OrensteinR11}]
There is a deterministic algorithm that, given integer parameters $ k \geq 0 $ and~$ \Delta \geq 1 $, a starting vertex~$ s $, and query access to a directed graph~$ G $, with $ O (\Gamma^{k + 2}) $ queries in general and with $ O (\Gamma^{k + 1} d) $ in graphs of maximum degree $ d $ detects whether there is a $k$-edge-out component containing~$ s $ of vertex size at most~$ \Gamma $ in~$ G $ and if so returns such a component. 
\end{theorem}

With a different motivation -- computing the $k$-edge connected components of a directed graph -- Chechik et al.~\cite{ChechikHILP17} developed a faster local cut-detection procedure for directed graphs.
The guarantees on the detected component are a bit weaker than those of Orenstein and Ron~\cite{OrensteinR11} as the component detected by the algorithm might be larger than $ \Delta $.
In the applications considered in this paper, this however is not an issue; as long as one is willing to pay the overhead in the complexity for detecting the component, the algorithms can also rely on the weaker guarantee.

\begin{theorem}[\cite{ChechikHILP17}]\label{lem:local procedure edge-out component Chechik}
There is a deterministic algorithm that, given integer parameters $ k \geq 0 $ and~$ \Delta \geq 1 $, a starting vertex~$ s $, and query access to a directed graph~$ G $, with $ O ((2 (k + 1))^{k+2} \Delta) $ queries returns a set of vertices $ U $ such that
(1) if there is a $ k $-edge-out component containing~$ s $ of edge size at most~$ \Delta $, then $ U \supseteq \{ s \} $ and
(2) if $ U \supseteq \{ s \} $, then $ U $ is a minimal $ k $-edge-out component containing~$ s $ of edge size at most $ O ((k + 1) \Delta) $.
\end{theorem}

Our approach is to follow the general algorithmic framework of Chechik et al.\ and to speed it up by using randomization.
In this way we reduce the dependence on $ k $ in the query complexity from exponential to polynomial.
Formally, we achieve the following guarantees.

\begin{theorem}\label{thm:local procedure edge-out component}
There is a randomized algorithm that, given integer parameters $ k \geq 0 $ and~$ \Delta \geq 1 $, a probability parameter $ p < 1 $, a starting vertex~$ s $, and query access to a directed graph~$ G $, with $ O ((k^2 + 1) (\Delta + k) \log (\tfrac{1}{1 - p})) $ queries returns a set of vertices $ U $ such that
(1) if there is a $ k $-edge-out component containing~$ s $ of edge size at most~$ \Delta $ in~$ G $, then $ U \supseteq \{ s \} $ with probability at least $ p $ and
(2) if $ U \supseteq \{ s \} $, then $ U $ is a minimal $ k $-edge-out component of edge size at most $ O((k + 1) (\Delta + k)) $ in~$ G $.
If every query to the graph can be performed in expected constant time, then the running time of the algorithm is $ O ((k^2 + 1) (\Delta + k) \log (\tfrac{1}{1 - p})) $.
\end{theorem}

We give a very general reduction that immediately makes this improvement carry over to detecting vertex cuts with the following guarantees.

\begin{theorem}\label{thm:local procedure vertex-out component}
There is a randomized algorithm that, given integer parameters $ k \geq 0 $ and~$ \Delta \geq 1 $, a probability parameter $ p < 1 $, a starting vertex~$ s $, and query access to a directed graph~$ G $, with $ O ((k^2 + 1) \Delta \log (\tfrac{1}{1 - p})) $ queries returns a set of vertices~$ U $ such that
(1) if there is a $ k $-vertex-out component containing~$ s $ of (symmetric) volume at most~$ \Delta $ in~$ G $, then $ U \supseteq \{ s \} $ with probability at least $ p $ and
(2) if $ U \supseteq \{ s \} $, then $ U $ is a $ k $-vertex-out component of (symmetric) volume at most $ O ((k + 1) (\Delta + k)) $ in~$ G $.
If every query to the graph can be performed in expected constant time, then the running time of the algorithm is $ O ((k^2 + 1) (\Delta + k) \log (\tfrac{1}{1 - p})) $.
\end{theorem}

This problem for vertex cuts has recently been studied by Nanongkai, Saranurak, and Yingchareonthawornchai~\cite{NanongkaiSY19} in the context of computing the global vertex connectivity.
They get the following guarantees.

\begin{theorem}[\cite{NanongkaiSY19}]
There is a deterministic algorithm that, given integer parameters $ k \geq 0 $ and~$ \Delta \geq 1 $, a starting vertex~$ s $, and a directed graph~$ G $ in adjacency-list representation, in time $ \tilde O (k \Delta^{3/2}) $ returns a set of vertices~$ U $ such that
(1) if there is a $ k $-vertex-out component containing~$ s $ of volume at most~$ \Delta $ in~$ G $, then $ U \supseteq \{ s \} $ with probability at least $ p $ and
(2) if $ U \supseteq \{ s \} $, then $ U $ is a $ k $-vertex-out component of volume at most $ O ((k + 1) \Delta) $ in~$ G $.
\end{theorem}

We remark that Nanongkai, Saranurak, and Yingchareonthawornchai also study an approximate version of this problem where the size of the vertex cut may exceed $ k $ by a factor of $ 1 + \epsilon $.

\subsubsection{Computing the Vertex Connectivity}

Computing the vertex connectivity $ \kappa $ of graph is a classic problem in graph algorithms.
Despite numerous research efforts, the conjectured linear-time algorithm~\cite{AhoHU74} has not yet been found.
In fact, even if $ \kappa $ is constant, no nearly-linear time algorithm has been given in the literature to date.
We give the first algorithm providing this guarantee.

For a long time, the state of the art was either a running time of $ \tilde O (m n) $ due to Henzinger, Rao, and Gabow~\cite{HenzingerRG00}, or a running time of $ \tilde O (n^\omega + n \kappa ^\omega) $ due to Cheriyan and Reif~\cite{CheriyanR94}, where $ \omega $ is the matrix-multiplication exponent for which currently $ \omega < 2.37287 $~\cite{Gall14a} is known.
In undirected graphs, the state of the art was either a running time of $ \tilde O (\kappa n^2) $ due to Henzinger, Rao, and Gabow~\cite{HenzingerRG00} or a running time of $ \tilde O (n^\omega + n \kappa^\omega) $ due to Linial, Lovász and Wigderson~\cite{LinialLW88}.
Very recently, Nanongkai, Saranurak, and Yingchareonthawornchai~\cite{NanongkaiSY19} improved upon some of these bounds by giving an algorithm with running time $ \tilde O (\kappa \cdot \min \{ m^{2/3} n, m^{4/3} \}) $ in directed graphs and an algorithm with running time $ \tilde O (\kappa^{7/3} n^{4/3} + m) $ in undirected graphs.
All of these algorithms are deterministic.
The fastest known deterministic algorithm by Gabow~\cite{Gabow06} has a running time of $ O ( \min \{ n^{3/4}, \kappa^{3/2} \} \cdot \kappa m + m n) $ in directed graphs and a running time of $ O ( \min \{ n^{3/4}, \kappa^{3/2} \} \cdot \kappa^2 n + \kappa n^2 ) $ in undirected graphs.
See~\cite{NanongkaiSY19} for a more thorough overview of prior work.
We significantly improve upon this state of the art, by giving an algorithm with running time $ \tilde O (\kappa^2 m + n) $ in directed graphs and an algorithm with running time $ \tilde O (\kappa^3 n + m) $ in undirected graphs.
Our algorithms are randomized Monte-Carlo algorithms that are correct with high probability.
Table~\ref{tab:results vertex connectivity} gives an overview over all of these results.

\begin{table}[htbp]
    \centering
    \begin{tabular}{c c c c}
    \textbf{Graph class} & \textbf{Running time} & \textbf{Deterministic} & \textbf{Reference} \\
    directed & $ O ( \min \{ n^{3/4}, \kappa^{3/2} \} \cdot \kappa m + m n) $ & yes & \cite{Gabow06} \\
    directed & $ \tilde O (m n) $ & no & \cite{HenzingerRG00} \\
    directed & $ \tilde O (n^\omega + n \kappa ^\omega) $ & no & \cite{CheriyanR94} \\
    directed & $ \tilde O (\kappa \cdot \min \{ m^{2/3} n, m^{4/3} \}) $ & no &  \cite{NanongkaiSY19} \\
    directed & $ \tilde O (\kappa^2 m + n) $ & no & \textbf{here} \\
    undirected & $ O ( \min \{ n^{3/4}, \kappa^{3/2} \} \cdot \kappa^2 n + \kappa n^2 ) $ & yes & \cite{Gabow06} \\
    undirected & $ \tilde O (\kappa n^2) $ & no & \cite{HenzingerRG00} \\
    undirected & $ \tilde O (n^\omega + n \kappa ^\omega) $ & no & \cite{LinialLW88} \\
    undirected & $ \tilde O (\kappa^{7/3} n^{4/3} + m) $ & no & \cite{NanongkaiSY19} \\
    undirected & $ \tilde O (\kappa^3 n + m) $ & no &  \textbf{here}
    \end{tabular}
    \caption{Comparison of algorithms for computing the vertex connectivity $ \kappa $}
    \label{tab:results vertex connectivity}
\end{table}

Our algorithms are the first to run in nearly linear time for a wide range of values of $ \kappa $ (whenever $ \kappa $ is polylogarithmic in $ n $).
Furthermore, in undirected graphs we give the current fastest solution as long as $ \kappa = \tilde o (n^{(\omega - 1) / 3}) $, i.e., roughly when $ \kappa \leq n^{0.45} $.
In directed graphs, the precise number depends on the density of the graph, but in any case we do give the fastest solution as long as $ \kappa = \tilde o (n^{(\omega - 2) / 2}) $, i.e., roughly when $ \kappa \leq n^{0.18} $.

\subsubsection{Computing the Maximal $k$-Edge Connected Subgraphs}

For a set $ C \subseteq $ of vertices, its induced subgraph~$ G [C] $ is a maximal $k$-edge connected subgraph of~$ G $ if $ G [C] $ is $k$-edge connected and no superset of $ C $ has this property.
The problem of computing all maximal $k$-edge connected subgraphs of $ G $ is a natural generalization of computing strongly connected components to higher edge connectivity.

For a long time, the state of the art for this problem was a running time of $ O (k m n \log{n}) $ for $ k > 2 $ and $ O (m n) $ for $ k = 2 $.
The first improvement upon this was given by Henzinger, Krinninger, and Loitzenbauer with a running of $ O (k^{O(1)} n^2 \log{n}) $ for $ k > 2 $ and $ O (n^2) $ for $ k = 2 $.
The second improvement was given by Chechik et al.~\cite{ChechikHILP17} with a running time of $ O ((2k)^{k + 2} m^{3/2} \log{n} + n) $ for $ k > 2 $ and $ O (m^{3/2}) $ for $ k = 2 $.
In undirected graphs, a version of the algorithm by Chechik et al.\ runs in time $ O ((2k)^{k + 2} m \sqrt{n} \log{n} + n) $ for $ k \geq 4 $ and in time $ O (m \sqrt{n} + n) $ for $ k \leq 3 $.
In this paper, we improve upon this by designing an algorithm that has expected running time $ O (k^{3/2} m^{3/2} \log n) $, reducing the dependence on $ k $ from exponential to polynomial.
We furthermore improve the running time for undirected graphs to $ O (k^4 n^{3/2} \log{n} + k m \log^2 n) $ thus improving both the dependence on $ k $ \emph{and} on $ m $.
Table~\ref{tab:results connected subgraphs} compares our results to previous results.

\begin{table}[htbp]
    \centering
    \begin{tabular}{c c c c c}
    \textbf{Parameter} & \textbf{Graph class} & \textbf{Running time} & \textbf{Deterministic} & \textbf{Reference} \\
    $ k = 2 $ & directed & $ O (m n) $ & yes & implied by~\cite{GabowT85} \\
    $ k \geq 3 $ & directed & $ O (k m n \log{n}) $ & yes & implied by~\cite{Gabow95} \\
    $ k = 2 $ & directed & $ O (n^2) $ & yes & \cite{HenzingerKL15} \\
    $ k \geq 3 $ & directed & $ O (k^{O(1)} n^2 \log{n}) $ & yes & \cite{HenzingerKL15} \\
    $ k = 2 $ & directed & $ O (m^{3/2} + n) $ & yes & \cite{ChechikHILP17} \\
    $ k \geq 3 $ & directed & $ O ((2k)^{k + 2} m^{3/2} \log{n} + n) $ & yes & \cite{ChechikHILP17} \\
    $ k \leq 3 $ & undirected & $ O (m \sqrt{n}) $ & yes & \cite{ChechikHILP17} \\
    $ k \geq 4 $ & undirected & $ O ((2k)^{k + 2} m \sqrt{n} \log{n} + n) $ & yes & \cite{ChechikHILP17} \\
    $ k \geq 2 $ & directed & $ O (k^{3/2} m^{3/2} \log n + n) $ & no & \textbf{here} \\
    $ k \geq 2 $ & undirected & $ O (k^4 n^{3/2} \log{n} + k m \log^2 n) $ & no & \textbf{here}
    \end{tabular}
    \caption{Comparison of algorithms for computing the maximal $k$-edge connected subgraphs}
    \label{tab:results connected subgraphs}
\end{table}

Note that another natural way of generalizing the concept of strongly connected components to higher edge connectivity is the following:
A $k$-edge connected component is a maximal subset of vertices such that any pair of distinct vertices is $k$-edge connected in $ G $.\footnote{Note that the $ k $ edge disjoint paths between a pair of vertices in a $k$-edge connected component might use edges that are not contained in the $k$-edge connected component. This is not allowed for maximal connected subgraphs.}
For a summary on the state of the art for computing the maximal $k$-edge connected subgraphs and components in both directed and undirected graphs, as well as the respective counterparts for vertex connectivity, we refer to~\cite{ChechikHILP17} and the references therein.

\subsubsection{Property Testing for Higher Connectivity}

After the seminal work of Goldreich and Ron~\cite{GoldreichR02}, property testing algorithms for $ k $-connectivity have been studied for various settings.
Table~\ref{tab:results unbounded degree} shows the state of the art in the unbounded-degree model and compares it with our results, whereas Table~\ref{tab:results bounded degree} shows the state of the art in the bounded-degree model and compares it with our results.
Observe that we subsume all prior results.
In particular, we solve the open problem of Goldreich and Ron~\cite{GoldreichR02} asking for a faster property testing algorithm for $ k $-edge connectivity in undirected graphs of bounded-degree.
Furthermore, we solve the open problem of Orenstein and Ron~\cite{OrensteinR11} asking for faster property testing algorithms for $ k $-edge connectivity and $ k $-vertex connectivity with polynomial dependence on $ k $ in the query complexity.

\begin{table}[htbp]
    \centering
    \begin{tabular}{l c c}
         & \textbf{State of the art} & \textbf{Here} \\
        undirected $k$-edge connectivity & $ \tilde O \left( \frac{k^4}{(\epsilon \davg)^4} \right) $~\cite{ParnasR02} & $ \tilde O \left( \frac{k^4}{(\epsilon \davg)^2} \right) $ \\ 
        directed $k$-edge connectivity & $ \tilde O \left( \left( \frac{c k}{\epsilon \davg} \right)^{k+1} \right) $~\cite{YoshidaI10,OrensteinR11} & $ \tilde O \left( \frac{k^4}{(\epsilon \davg)^2} \right) $ \\
        undirected $k$-vertex connectivity & $ \tilde O \left( \left( \frac{c k}{\epsilon \davg} \right)^{k+1} \right) $~\cite{OrensteinR11} & $ \tilde O \left( \frac{k^5}{(\epsilon \davg)^2} \right) $ \\
        directed $k$-vertex connectivity & $ \tilde O \left( \left( \frac{c k}{\epsilon \davg} \right)^{k+1} \right) $~\cite{OrensteinR11} & $ \tilde O \left( \frac{k^5}{(\epsilon \davg)^2} \right) $
    \end{tabular}
    \caption{Comparison of property testing algorithms for higher connectivity in unbounded-degree graphs}
    \label{tab:results unbounded degree}
\end{table}

\begin{table}[htbp]
    \centering
    \begin{tabular}{l c c}
         & \textbf{State of the art} & \textbf{Here} \\
        undirected $2$-edge connectivity & $ O \left( \frac{\log^2 (\frac{1}{\epsilon d})}{\epsilon} \right) $~\cite{GoldreichR02} & $ O \left( \frac{\log^2 (\frac{1}{\epsilon d})}{\epsilon} \right) $ \\
        undirected $3$-edge connectivity & $ O \left( \frac{\log (\frac{1}{\epsilon d})}{\epsilon^2 d} \right) $~\cite{GoldreichR02} & $ O \left( \frac{\log^2 (\frac{1}{\epsilon d})}{\epsilon} \right) $ \\
        undirected $k$-edge connectivity & $ \tilde O \left( \frac{k^3}{\epsilon^{3 - \frac{2}{k}} d^{2 - \frac{2}{k}}} \right) $~\cite{GoldreichR02} & $ \tilde O \left( \frac{k^3}{\epsilon} \right) $ \\
        directed $k$-edge connectivity & $ \tilde O \left( \left( \frac{c k}{\epsilon d} \right)^k d \right) $~\cite{YoshidaI10} & $ \tilde O \left( \frac{k^3}{\epsilon} \right) $ \\
        undirected $k$-vertex connectivity & $ \tilde O \left( \left( \frac{c k}{\epsilon d} \right)^k d \right) $~\cite{YoshidaI12} & $ \tilde O \left( \frac{k^3}{\epsilon} \right) $ \\
        directed $k$-vertex connectivity & $ \tilde O \left( \left( \frac{c k}{\epsilon d} \right)^k d \right) $~\cite{OrensteinR11} & $ \tilde O \left( \frac{k^3}{\epsilon} \right) $
    \end{tabular}
    \caption{Comparison of property testing algorithms for higher connectivity in bounded-degree graphs}
    \label{tab:results bounded degree}
\end{table}

Note that the query complexity of the $k$-edge connectivity tester for undirected bounded-degree graphs in Theorem~3.1 of~\cite{GoldreichR02} is stated as $ O \left( \frac{k^3 \log^2 (1 / (\epsilon d))}{\epsilon^{3 - \frac{2}{k}} d^{2 - \frac{2}{k}}} \right) $.
However, to the best of our understanding, this bound only applies for $ (k-1) $-connected graphs.
Following Algorithm~3.18 in~\cite{GoldreichR02}, the query complexity in arbitrary bounded-degree graphs should therefore be $ O \left( \frac{k^3 \log^2 (\frac{k}{\epsilon d})}{\epsilon^{3 - \frac{2}{k}} d^{2 - \frac{2}{k}}} \right) $.
While such a deviation is certainly marginal and often hidden with good reason in the $ \tilde O (\cdot) $-notation, we do report it here to make clear that no further arguments than the combinatorial ones of Orenstein and Ron~\cite{OrensteinR11} and the algorithmic ones in this paper are needed to obtain the state of the art.

\paragraph{Independent Work}
In follow-up work to~\cite{NanongkaiSY19}, Nanongkai, Saranurak, and Yingchareonthawornchai have independently claimed results with the same guarantees as ours for locally computing a bounded-size edge or vertex cut, for computing the vertex connectivity of a directed or an undirected graph, and for computing the maximal $k$-edge connected subgraphs of a directed graph~\cite{Saranurak19}.

\section{Local Cut Detection}\label{sec:local cut detection}

In this section we first prove Theorem~\ref{thm:local procedure edge-out component} by giving a fast local procedure for detecting a $ k $-edge-out component of edge size at most $ \Delta $.
We then obtain an analogous statement for vertex connectivity.

\subsection{Detecting Bounded-Size Edge Cuts}

We exploit that a $ k $-edge out component has at most $ k $ edge-disjoint paths leaving the component.
Once we ``block'' these paths, we will not be able to leave the component in any other way and may conclude that the set of vertices reachable from $ s $ is a $ k $-edge out component.
In particular, we try to find these paths in an iterative manner using only simple depth-first searches (DFS).
In principle, we can hope to keep each DFS ``local'' by exploring only the neighborhood of $ s $ as the component has at most $ \Delta $ edges.
However, since we do not know the component in advance, we do not know which vertices visited by each DFS are outside of the component.
Our main idea is to perform each DFS up to a ``budget'' for processing $ \Omega (k (\Delta + k)) $ edges and to then sample one of the processed edges uniformly at random.\footnote{We consider the variant of DFS using a stack in which the stack initially contains the starting vertex and in each iteration a vertex is popped from the stack and visited, which means that all its outgoing edges are processed by pushing their other endpoints on the stack.}
We add the path from $ s $ to the tail of the sampled edge in the current DFS tree to our set of chosen paths.
If there is a $k$-edge-out component containing $ s $ of edge size at most $ \Delta $, then there are only $ \Delta + k $ edges whose tails are inside the component.
Thus, the endpoint of our new path is lying inside of the component with probability only $ O (\tfrac{\Delta + k }{k \cdot (\Delta + k)}) = O (\tfrac{1}{k}) $.
We repeat this process $ k $~times, blocking the edges of the paths chosen so far in each DFS. 
Thus, the probability of \emph{one} of the tails of a sampled edge lying inside of the component is at most $ O (k \cdot \tfrac{1}{k}) = O (1) $ by the union bound.
In this way, we obtain $ k $~paths leaving the component with constant probability and we can verify that an additional DFS not using one of these paths cannot reach more than the $ \Delta $ edges of the component.
By a standard boosting approach we can increase the success probability at the cost of repeating this algorithm multiple times.

As described so far, this approach still has the problem that there might be no small $k$-edge-out component containing $ s $ at all and the final DFS might simply fail to process more than $ \Delta $ edges because our choice of paths blocked some relevant edges.
In particular, each of our chosen paths might leave and re-renter the component multiple times, as we only have a sufficiently high probability that the \emph{endpoint} of the path is outside of the component.
To avoid such a situation, we do not literally block the edges of the paths found so far by removing them from the graph.
Instead, we use the augmenting paths framework of Chechik et al.~\cite{ChechikHILP17} that after each iteration reverses the edges of our chosen path, similar to residual graph constructions in maximum flow algorithms.

In more detail, our algorithm works as follows:
First, we perform up to $ k $ depth-first searches up to a budget of $ 2 k (\Delta + k) $ edges.
Consider the $ i $-th DFS.
If the DFS is completed before the budget on the number of edges is exceeded (i.e., if the DFS processes less than $ 2 k (\Delta + k) $ edges), we return the set of vertices found in the DFS as a $k$-edge-out component.
Otherwise, we sample one edge $ (u_i, v_i) $ processed by the DFS uniformly at random and let $ \pi_i $ be the path from $ s $ to the tail $ u_i $ of this edge in the DFS tree.
A special case arises if the edge $ (u_i, v_i) $ is the reversal of the edge $ (v_i, u_i) $.
If this happens, then we let $ \pi_i $ be the path from $ s $ to $ v_i $, the tail of the original edge in $ E $, in the DFS tree.
We then reverse the edges on $ \pi_i $ in the graph and start the next DFS.
Finally, we perform another DFS up to a budget of $ \Delta + 1 $ edges.
Again, if the DFS is completed before the budget on the number of edges is exceeded, we return the set of vertices found in the DFS as a $k$-edge-out component.
Otherwise, we return the empty set to indicate that no $ k $-out-edge component of edge size at most $ \Delta $ has been found.
The pseudocode of this procedure is given in Algorithm~\ref{alg:local procedure}.

\begin{algorithm2e}
\caption{Local procedure for detecting a $ k $-edge-out component}\label{alg:local procedure}
\SetKwFunction{DetectComponent}{DetectComponent}
\SetKwFunction{DetectComponentParameterized}{DetectComponentParam}
\tcp{The procedures have query access to the input graph $ G $. They try to detect a $k$-edge-out component of edge size at most~$ \Delta $ containing~$ s $.}
\Procedure{\DetectComponent{$ s $, $ k $, $ \Delta $, $ G = (V, E) $}}{
    $ E_0 \gets E $\;
    $ G_0 \gets (V, E_1) $\;
    \For{$ i = 0 $ \KwTo $ k - 1 $}{
        Perform depth-first search from $ s $ in $ G_i $ processing up to $ 2 k (\Delta + k) $ edges~$ F_i $ and visiting vertices~$ U_i $ with resulting partial DFS tree~$ T_i $\;
        \If{$ |F_i| < 2 k (\Delta + k) $}{
            \Return $ U_i $\;
        }
        Pick one edge $ e_i = (u_i, v_i) $ from $ F_i $ uniformly at random\;
        \eIf(\tcp*[f]{Check if $ e_i $ has been reversed}){$ e_i \in E $}{
            Let $ \pi_i $ be the path from $ s $ to $ u_i $ in $ T_i $\;
        }{
            Let $ \pi_i $ be the path from $ s $ to $ v_i $ in $ T_i $\;
        }
        \tcp{Reverse the edges of $ \pi_i $}
        $ E_{i + 1} = E_i \setminus \{ (u, v) : (u, v) \in \pi_i \} \cup \{ (v, u) : (u, v) \in \pi_i \} $\;\label{line:reverse edges on path}
        $ G_{i + 1} = (V, E_{i + 1}) $\;
    }
    Perform depth-first search from $ s $ in $ G_k $ processing up to $ \Delta + 1 $ edges~$ F_k $ and visiting vertices~$ U_k $$ F_k $\;\label{line:final DFS}
    \eIf{$ |F_k| \leq \Delta $}{
        \Return $ U_k $\;
    }{
        \Return $ \emptyset $\;\label{line:algorithm returns 0}
    }
}

\BlankLine

\tcp{The additional parameter $ p $ controls the success probability}
\Procedure{\DetectComponentParameterized{$ s $, $ k $, $ \Delta $, $ G $, $ p $}}{
    \For{$ \lceil \log (\frac{1}{1 - p}) \rceil $ times}{\label{line:repeat procecdure}
        $ U \gets $ \DetectComponent{$ s $, $ k $, $ \Delta $, $ G $}\;
        \If{$ U \neq \emptyset $}{
            \Return $ U $\;
        }
    }
    \Return $ \emptyset $\;
}
\end{algorithm2e}

In the following, we prove Theorem~\ref{thm:local procedure edge-out component} by a series of lemmas.
We start with the following lemmas from Chechik et al.~\cite{ChechikHILP17}; we give the proofs for completeness using our own notation.
\begin{lemma}[\cite{ChechikHILP17}]\label{lem:out edge reduction}
Let $ S \subseteq V $ be a set of vertices containing~$ s $ and let $ T = V \setminus S $.
Assume that $ T $ contains $ \ell $ of the endpoints of the paths $ \pi_0, \ldots, \pi_i $ found in Procedure \DetectComponent for some $ \ell \leq i + 1 $ and any $ i \leq k - 1 $.
Then in $ G_{i+1} $ there are $ \ell $ fewer edges from $ S $ to $ T $ than in $ G $.
\end{lemma}

\begin{proof}
Consider the (multi-)graph $ G' $ that is obtained from $ G $ by contracting the vertices of $ S $ to a single vertex $ s' $ and the vertices of $ T $ to a single vertex $ t' $.
Applying the contraction to the paths $ \pi_j $, we obtain for each $ 0 \leq j \leq i $ a set of edges $ P_j $ between $ s' $ and $ t' $ that represents the contraction of $ \pi_j $, where we keep the direction of edges as in $ \pi_j $.
Let $ G_0' = G $ and for $ 0 \leq j \leq i $ let $ G_{j+1}'$ be the (multi-)graph obtained from $ G_j' $ by reversing the edges of $ P_j $.
Note that the graph $ G_j' $ can also be obtained from $ G_j $ by contracting $ S $ and $ T $, respectively.
By definition, the graphs $ G_j' $ differ from $ G' $ only in the direction of the edges between $ s' $ and~$ t' $.
Further we have that if $ \pi_j $ ends at a vertex of $ T $ (case 1), then in $ P_j $ the number of edges from $ s' $ to~$ t' $ is one more than the number of edges from $ t' $ to~$ s' $; in contrast, if $ \pi_j $ ends at a vertex of $ S $ (case 2), then in $ P_j $ there are as many edges from $ s' $ to~$ t' $ as from $ t $ to $ s $.
In case 1 the number of edges from $ s' $ to $ t' $ in $ G_{j+1}' $ is one lower than in $ G_j' $, while in case 2 the number of edges from $ s' $ to $ t' $ is the same in $ G_j' $ and $ G_{j+1}' $.
Let $ 0 \leq \ell \leq i + 1 $ be the number of paths among $ \pi_0, \ldots, \pi_i $ that end in $ T $.
We have that the number of edges from $ s' $ to $ t' $ in $ G_{i + 1}' $, and therefore the number of edges from $ S $ to $ T $ in $ G_{i + 1} $, is equal to the number of paths from $ S $ to $ T $ in $ G $ minus $ \ell $.
\end{proof}

\begin{lemma}[Implicit in \cite{ChechikHILP17}]\label{lem:reachable vertices are component}
If Procedure \DetectComponent returns $ U_i $ for some $ 0 \leq i \leq k $, then $ U_i $ is a minimal $ i $-edge-out component containing~$ s $ in $ G $.
\end{lemma}

\begin{proof}
Let $ S $ be the set of vertices reachable from $ s $ in $ G_i $ and let $ T $ be $ V \setminus S $.
Observe that the set $ U_i $ is returned only if the $i$-th DFS has been completed without being stopped because of exceeding its ``edge budget''.
Therefore $ U_i = S $.

If $ i = 0 $, then $ S $ is the set of vertices reachable from $ s $ in $ G $, which trivially is a minimal $0$-edge-out component of $ G $ containing $ s $.
Consider now the case $ i \geq 1 $.
By the definition of $ S $, $ S $ contains $ s $ and there are no edges from $ S $ to $ T $ in $ G_i $.
Thus by Lemma~\ref{lem:out edge reduction} the number of edges from $ S $ to~$ T $ in~$ G $ is equal to the number $ \ell $ of paths among $ \pi_0, \ldots, \pi_{i-1} $ that end in~$ T $.
Thus, $ S $ has $ \ell \leq i $ outgoing edges in $ G $, i.e., $ S $ is an $ i $-edge-out component of~$ G $.
It remains to prove the minimality of $ S $, i.e., to show that $ S $ does not contain a proper subset~$ \hat{S} $ that contains $ s $ and has $ \ell $ or less outgoing edges.
Assume by contradiction that such a set $ \hat{S} $ exists and let $ \hat{T} = V \setminus \hat{S} $.
By $ \hat{T} \supset T $, at least $ \ell $ of the paths among $ \pi_0, \ldots, \pi_{i-1} $ end in $ \hat{T} $.
Thus, by Lemma~\ref{lem:out edge reduction}, there are no edges from $ \hat{S} $ to $ \hat{T} $ in $ G_i $.
This implies that the set of vertices reachable from $ s $ in $ G_i $ is $ \hat{S} $, contradicting the assumption that $ S \supset \hat{S} $ is this set.
\end{proof}

\begin{lemma}[Implicit in \cite{ChechikHILP17}]\label{lem:final reachable set is component}
Let $ C $ be a minimal $k$-edge-out component containing~$ s $ of edge size at most~$ \Delta $ and assume for each $ 0 \leq i \leq k-1 $ that at least $ 2 k (\Delta + k) $ edges are reachable from $ s $ in $ G_i $ and that $ \pi_i $ ends in $ V \setminus C $.
Then Procedure \DetectComponent returns $ U_k \supseteq \{ s \} $.
\end{lemma}

\begin{proof}
By definition, there are at most $ k $ edges from $ C $ to $ V \setminus C $ in $ G $.
As by assumption, all paths $ \pi_0, \ldots, \pi_{k-1} $ end in $ V \setminus C $, there are no edges from $ C $ to $ V \setminus C $ in $ G_k $ by Lemma~\ref{lem:out edge reduction}.
As $ C $ contains $ s $ and $ C $ has edge size at most $ \Delta $, the number of vertices reachable from~$ s $ in~$ G_k $ is at most~$ \Delta $.
Thus, the DFS in Line~\ref{line:final DFS} of Procedure \DetectComponent traverses all vertices reachable from $ s $ in $ G_k $ and a non-empty set $ U_k $ containing at least $ s $ is returned.
\end{proof}

\begin{lemma}\label{lem:probability of paths ending in component}
Let $ C $ be a minimal $k$-edge-out component containing~$ s $ of edge size at most~$ \Delta $ and assume for each $ 0 \leq i \leq k-1 $ that at least $ 2 k (\Delta + k) $ edges are reachable from $ s $ in $ G_i $.
Then the probability that there is some $ 0 \leq i \leq k - 1 $ such that $ \pi_i $ ends in $ C $ is at most $ \tfrac{1}{2} $.
\end{lemma}

\begin{proof}
First, fix some $ i $.
The $i$-th sampled edge $ e_i = (u_i, v_i) $ is either contained in~$ E $ or its reverse edge $ (v_i, u_i) $ is.
In the first case $ \pi_i $ ends in $ C $ if and only if the tail $ u_i $ of the $ e_i $ is contained in $ C $ and in the second case $ \pi_i $ ends in $ C $ if and only if the tail $ v_i $ of the reverse of $ e_i $ is contained in~$ C $.
As $ C $ has edge size at most~$ \Delta $ and there are at most $ k $~edges leaving~$ C $, there are at most $ \Delta + k $ edges of~$ E $ whose tail is contained in~$ C $.
By the assumption we have $ | F_i | \geq 2 k (\Delta + k) $, i.e., $ e_i $ was sampled from a set of $ 2 k (\Delta + k) $ distinct edges.
Thus, the probability that $ \pi_i $ ends in $ C $ is at most $ \tfrac{\Delta + k}{2 k (\Delta + k)} = \tfrac{1}{2 k} $.
Now by the union bound, the probability that for least one $ 0 \leq i \leq k-1 $ the path $ \pi_i $ ends in $ C $ is at most $ k \cdot \tfrac{1}{2 k} = \tfrac{1}{2} $.
\end{proof}

\begin{lemma}\label{lem:procedure find component}
If Procedure \DetectComponent returns $ U \supseteq \{ s \} $, then $ U $ is a minimal $k$-edge-out component of edge size at most $ \max (2 k (\Delta + k), \Delta) $ in~$ G $.
If $ G $ has a $k$-edge-out component containing~$ s $ of edge size at most $ \Delta $, then Procedure \DetectComponent returns $ U \supseteq \{ s \} $ with probability at least $ \tfrac{1}{2} $.
The query complexity of Procedure \DetectComponent is $ O ((k^2 + 1) (\Delta + k)) $.
\end{lemma}

\begin{proof}
To prove the correctness claims, observe first that $ s $ is always reachable from itself and thus the only possibility for the procedure to return a set not containing $ s $ is in Line~\ref{line:algorithm returns 0}  (where it returns $ \emptyset $).
Now the first correctness claim follows from Lemma~\ref{lem:reachable vertices are component} and the fact that all sets returned by the algorithm have edge size at most $ 2 k (\Delta + k) $ if $ k \geq 1 $ and at most $ \Delta $ if $ k = 0 $.
It remains to show that Line~\ref{line:algorithm returns 0} is executed with probability at most $ \tfrac{1}{2} $.
A precondition for this to happen is that for every $ 0 \leq i \leq k - 1 $ the number of edges reachable from~$ s $ in~$ G_i $ is at least $ 2 k (\Delta + k) $.
Let $ C $ be a $k$-edge-out component containing~$ s $ of edge size at most $ \Delta $ in $ G $.
By Lemma~\ref{lem:probability of paths ending in component} the probability that every path $ \pi_i $ (for $ 0 \leq i \leq k-1 $) ends in $ V \setminus C $ is at most $ \tfrac{1}{2} $.
If that is the case, then the procedure returns $ U_k \supseteq \{ s \} $ by Lemma~\ref{lem:final reachable set is component}.
Thus, the probability of the procedure returning a set containing $ s $ is at least~$ \tfrac{1}{2} $.

To bound the query complexity, observe that the procedure performs at most one depth-first search up to $ \Delta + 1 $ edges and $ k $ depth-first searches up to $ 2 k (\Delta + k) $ edges.
Thus, the total number of queries is bounded by $ O ((k^2 + 1) (\Delta + k)) $.
\end{proof}

\begin{lemma}
If Procedure \DetectComponentParameterized returns $ U \supseteq \{ s \} $, then $ U $ is a minimal $k$-edge-out component of edge size at most $ \max (2 k (\Delta + k), \Delta) $ in $ G $.
If $ G $ has a $k$-edge-out component containing~$ s $ of edge size at most~$ \Delta $, then Procedure \DetectComponentParameterized returns $ U \supseteq \{ s \} $ with probability at least $ p $.
The query complexity of Procedure \DetectComponentParameterized is $ O ((k^2 + 1) (\Delta + k) \log (\tfrac{1}{1 - p})) $.
\end{lemma}

\begin{proof}
The first part of the lemma directly follows from the first part of Lemma~\ref{lem:procedure find component}.
To prove the second part of the lemma, assume that $ G $ has a $k$-edge-out component containing~$ s $ of edge size at most~$ \Delta $.
The probability that a single call of Procedure \DetectComponent returns a set not containing $ s $ is at most $ \tfrac{1}{2} $ by Lemma~\ref{lem:procedure find component}.
Thus, the probability that $ \lceil \log (1 - p) \rceil $ independent calls of Procedure \DetectComponent each return a set not containing $ s $ is at most $ (\tfrac{1}{2})^{\lceil \log (\frac{1}{1 - p}) \rceil} \leq 1 - p $.
It follows that the probability that \DetectComponentParameterized returns a set containing $ s $ is at least $ 1 - (1 - p) = p $.

The bound on the query complexity and the running time directly follows from Lemma~\ref{lem:procedure find component}
\end{proof}

We finally argue that Procedure \DetectComponentParameterized can be implemented to run in time $ O ((k^2 + 1) (\Delta + k) \log (\tfrac{1}{1 - p})) $.
In Procedure \DetectComponent, we store all edges queried so far in hash tables, with one hash table per vertex containing its incident edges.
Whenever we reverse the direction of some edge $ (u, v) $ in Line~\ref{line:reverse edges on path}, we do so by replacing $ (u, v) $ with $ (v, u) $ in the hash tables of $ u $ and $ v $.
Recall that by assumption each query takes constant time in expectation and there are $ O ((k^2 + 1) (\Delta + k)) $ queries.
As additionally each path~$ \pi_i $ has length $ O (k (\Delta + k)) $ and there are at most $ k $ such paths for which the edges are reversed, the number of operations to the hash tables is $ O ((k^2 + 1) (\Delta + k)) $.
This gives an expected running time of $ O ((k^2 + 1) (\Delta + k)) $ for Procedure \DetectComponent.
By being a little more careful we can also get a worst-case bound by tolerating a slightly larger error probability: we stop the algorithm and return~$ \emptyset $ whenever its running time so far exceeds the expected bound by a factor of $ 4 $.
By Markov's bound this happens with probability at most $ \tfrac{1}{4} $.
This decreases the probability of Procedure \DetectComponent to return $ U \neq \emptyset $ from $ \tfrac{1}{2} $ to $ \tfrac{1}{4} $.
We can account for this increase by increasing the number of repetitions in Line~\ref{line:repeat procecdure} of procedure \DetectComponentParameterized from $ \lceil \log (\frac{1}{1 - p}) \rceil $ to $ \lceil \log_{4/3} (\frac{1}{1 - p}) \rceil $.
Thus, the running time of Procedure \DetectComponentParameterized will be $ O ((k^2 + 1) (\Delta + k) \log (\tfrac{1}{1 - p})) $ and Theorem~\ref{thm:local procedure edge-out component} follows.

\subsection{Detecting Bounded-Size Vertex Cuts}

In the following, we prove Theorem~\ref{thm:local procedure vertex-out component} by reducing local vertex cut detection to local edge cut detection.
To do this, we modify a well-known reduction that has previously been used for computing the local vertex connectivity of a pair of vertices by performing a maximum flow computation~\cite{Even75}.

Given a directed graph $ G = (V, E) $ containing vertex~$ s $, define the graph $ G'_s = (V'_s, E'_s) $ as follows:
\begin{itemize}
\item For every vertex $ v \in V \setminus \{ s \} $, $ G'_s $ contains two vertices $ v_{\mathrm{in}} $ and $ v_{\mathrm{out}} $ and additionally a vertex $ s_{\mathrm{in}} = s_{\mathrm{out}}  $, i.e.,
\begin{equation*}
V'_s = \{ v_{\mathrm{in}} : v \in V \} \cup \{ v_{\mathrm{out}} : v \in V \}
\end{equation*}
where only $ s_{\mathrm{in}} $ and $ s_{\mathrm{out}} $ are identical.
\item For every vertex $ v \in V $, $ v_{\mathrm{in}} $ gets all incoming edges of $ v $, $ v_{\mathrm{out}} $ gets all outgoing edges of $ v $, and there additionally is an edge from $ v_{\mathrm{in}} $ to $ v_{\mathrm{out}} $, i.e.,
\begin{equation*}
E'_s = \{ (v_{\mathrm{out}}, w_{\mathrm{in}}) : (v, w) \in E \} \cup \{ (v_{\mathrm{in}}, v_{\mathrm{out}}): v \in V \} \, .
\end{equation*}
\end{itemize}
Note that we do not explicitly have to modify the input graph $ G $ (to which we have query access) in algorithmic applications.
Any algorithm running on~$ G'_s $ can on-the-fly decide for each edge whether it is from the first set in~$ E'_s $ or from the second set in~$ E'_s $.
In the first case, the edge can be queried from~$ G $, and in the second case the edge can be created by the algorithm instantly.

Let $ C \subseteq V $ be a subset of edges in $ G $.
Let $ B = \{ v \in V \setminus C | \exists u \in C : (u, v) \in E \} $ be its set of boundary vertices.
Recall that the symmetric volume of $ C $ in $ G $ is defined as
\begin{equation*}
    \vol^* (C) = | E (C, C) | + | E (C, B)| + | E (V \setminus C, C) | \, .
\end{equation*}
Let $ C' \subseteq V'_s $ be a subset of edges in $ G'_s $ and let $ B' = \{ v \in V'_s \setminus C' | \exists u \in C' : (u, v) \in E'_s \}  $ be its set of boundary vertices.
Furthermore, define the \emph{interior of $ C' $} as $ I' = C' \setminus \{ v_{\mathrm{in}} \in V'_s : v_{\mathrm{out}} \in V'_s \setminus C' \} $ and define the \emph{restricted symmetric volume of $ C' $} in $ G'_s $ as
\begin{equation*}
    \vol' (C') = | E (C', C') | + | E (C', B')| + | E (V'_s \setminus C', I') | \, .
\end{equation*}

We now present two lemmas that formally express the tight connections between $k$-vertex-out components containing $ s $ in $ G $ and $k$-vertex-out components containing $ s_{\mathrm{out}} $ in $ G'_s $.

\begin{lemma}\label{lem:vertex component implies edge component}
If there is a $k$-vertex-out component $ C $ containing~$ s $ in~$ G $ with $ \vol (C) = \Delta $ ($ \vol^* (C) = \Delta $), then there is a $k$-edge-out component $ C' $ containing~$ s_{\mathrm{out}} $ in~$ G'_s $ with $ \vol (C') \leq 3 \Delta $ ($ \vol' (C') \leq 3 \Delta $).
\end{lemma}

\begin{proof}
Let $ B = \{ v \in V \setminus C | \exists u \in C : (u, v) \in E \} $ denote the at most $ k $ boundary vertices of~$ C $, i.e., those vertices with incoming edges from~$ C $.
Define $ C' $ as the component that is cut at the edges connecting $ v_{\mathrm{in}} $ to $ v_{\mathrm{out}} $ for the boundary vertices, i.e.,
\begin{equation*}
C' := \{ v_{\mathrm{out}} : v \in C \} \cup \{ v_{\mathrm{in}} : v \in C \} \cup \{ v_{\mathrm{in}} : v \in B \} \, .
\end{equation*}
Now let $ e $ be an edge leaving $ C' $ in $ G'_s $ and suppose $ e $ is of the form $ (v_{\mathrm{out}}, w_{\mathrm{in}}) $.
Then $ (v, w) \in E $, and, by the definition of $ C' $, $ v \in C $ and $ w \in V \setminus ( C \cup B ) $, which contradicts the fact that~$ B $ contains all boundary vertices.
Therefore, every edge leaving $ C' $ in $ G'_s $ must be of the form $ (v_{\mathrm{in}}, v_{\mathrm{out}}) $ which, by the construction of $ C' $ is only possible for $ v \in B $.
Thus, there are at most $ |B| \leq k $ edges leaving $ C'$.

We now show that $ \vol' (C') \leq 3 \Delta $.
The proof that $ \vol (C') \leq 3 \Delta $ would follow a similar, but simpler, argument.
In particular, we will show that for every edge $ e' \in  E (C' C') \cup E (C', B') \cup E (V'_s \setminus C', I') $ we either have $ e' = (v_{\mathrm{in}}, v_{\mathrm{out}}) $ for some $ v \in C \cup B $ or $ e' = (u_{\mathrm{out}}, v_{\mathrm{in}}) $ for some $ (u, v) \in E (C, C) \cup E (C, B) \cup E (V \setminus C, C) $.
It then follows that
\begin{align*}
    \vol' (C') &= | E (C', C') \cup E (C', B') \cup E (V'_s \setminus C', I') | \\
    &\leq | E (C, C) \cup E (C, B) \cup E (V \setminus C, C) | + | C \cup B | \\
    &\leq \vol^* (C) + 2 \vol^* (C) \\
    &\leq 3 \Delta
\end{align*}
as desired

Let $ e' \in E (C' C') \cup E (C', B') \cup E (V'_s \setminus C', I') $ and assume that $ e' = (u_{\mathrm{out}}, v_{\mathrm{in}}) $ for some $ (u, v) $ for some $ (u, v) \in E $.
If $ u_{\mathrm{out}} \in C' $, then by the definition of $ C' $ it must be the case that $ u \in C $ and therefore $ (u, v) \in E (C, C) \cup E (C, B) $.
Otherwise, we have $ (u_{\mathrm{out}}, v_{\mathrm{in}}) \in (V'_s \setminus C', I') $.
By the definition of $ C' $, $ u \in V \setminus C $, and by the definition of $ I' $, $ v_{\mathrm{in}} \in I' $ implies $ v_{\mathrm{out}} \in C' $ and thus $ v \in C $.
It follows that $ (u, v) \in E (V \setminus C, C) $.
\end{proof}

\begin{lemma}\label{lem:edge component implies vertex component}
Let $ C' $ be a minimal $k$-edge-out component containing~$ s_{\mathrm{out}} $ in~$ G'_s $.
Then $ C = \{ v \in V \mid v_{\mathrm{out}} \in C' \} $ is a $k$-vertex-out component containing~$ s $ in~$ G $ with $ \vol^* (C) \leq \vol' (C') $.
\end{lemma}

\begin{proof}
In the proof we will use counting arguments that implicitly consider the one-to-one mappings that map each vertex $ v \in V $ to the edge $ (v_{\mathrm{in}}, v_{\mathrm{out}}) \in E'_s $ and each edge $ (u, v) \in E $ to the edge $ (u_{\mathrm{out}}, v_{\mathrm{in}}) \in E'_s $.

Let $ B = \{ v \in V \setminus C \mid \exists u \in V : (u, v) \in E \} $ be the boundary vertices of $ C $.
We will show that to every vertex $ v \in B $ we can uniquely assign one edge $ e' \in E'_s (C', V'_s \setminus C') $, which implies $ | B  | \leq | E'_s (C', V \setminus C') | $.
As $ C' $ is a $k$-edge-out component we have $ | E'_s (C', V'_s \setminus C') | \leq k $.
Therefore $ | B | \leq k $ which means that $ C $ is a $k$-vertex-out component.

Let $ v \in B $.
By the definition of $ C $, we have $ v_{\mathrm{out}} \in V'_s \setminus C' $.
If $ v_{\mathrm{in}} \in C' $, then $ e' = (v_{\mathrm{in}},  v_{\mathrm{out}}) \in E'_s (C', V'_s \setminus C') $.
Otherwise, if $ v_{\mathrm{in}} \in V'_s \setminus C' $, the argument is as follows:
As $ v \in B $, there must be at least one edge $ (u, v) \in E (C, V \setminus C) $.
By the definition of $ C' $, $ u_{\mathrm{out}} \in C' $ and thus $ e' = (u_{\mathrm{out}}, v_{\mathrm{in}}) \in E'_s (C', V'_s \setminus C') $.
In both cases, $ v $ can be uniquely assigned to $ e' $ as desired.

Finally, we show that $ \vol^* (C) \leq \vol' (C') $.
The proof that $ \vol (C) \leq \vol (C') $ would follow a similar, but simpler, argument.
Consider some edge $ e = (u, v) \in E (C, C) \cup E (C, B) \cup E (V \setminus C, C) $.
If $ u_{\mathrm{out}} \in C' $, then $ u \in C $ by the definition of $ C $, and thus trivially $ e' = (u_{\mathrm{out}}, v_{\mathrm{in}}) \in E (C', C') \cup E (C', B') $.
If $ u_{\mathrm{out}} \notin C' $, then $ u \in V \setminus C $ and it therefore must be the case that $ v \in C $ and thus $ v_{\mathrm{out}} \in C' $ and, by the definition of $ I' $, $ v_{\mathrm{out}} \in I' $.

Now suppose that $ v_{\mathrm{in}} \in V'_s \setminus C' $.
Then the edge $ (v_{\mathrm{in}}, v_{\mathrm{out}}) $ (which exists by the definition of $ G'_s $) goes from $ V'_s \setminus C' $ to $ C' $.
However, any path to $ v_{\mathrm{out}} $ must have $ (v_{\mathrm{in}}, v_{\mathrm{out}}) $ as its last edge.
Let $ R \subseteq \{ v_{\mathrm{in}} \} $ be the subset of $ C' $ that contains all vertices reachable from $ v_{\mathrm{in}} $ using only vertices of $ C' $.
Then $ C' \setminus R $ is a $k$-edge-out component with the same number of outgoing edges as $ C' $, which contradicts the minimality of $ C' $.
It therefore must be the case that $ v_{\mathrm{in}} \in C' $ and thus $ e' = (u_{\mathrm{out}}, v_{\mathrm{in}}) \in E (V'_s \setminus C', I') $.

We have shown that for every edge $ e = (u, v) \in E (C, C) \cup E (C, B) \cup E (V \setminus C, C) $ we have $ e' =(u_{\mathrm{out}}, v_{\mathrm{in}}) \in E (C', C) \cup E (C', B') \cup E (V'_s \setminus C', I') $.
Therefore, $ \vol^* (C) \leq \vol' (C') $.
\end{proof}

We now obtain the local algorithm of Theorem~\ref{thm:local procedure vertex-out component} for detecting vertex cuts as follows by running a variant of the procedure of Theorem~\ref{thm:local procedure edge-out component} on $ G'_s $ with starting vertex $ s_{\mathrm{out}} $ and parameter~$ \Delta' = 3 \Delta $, which returns some set of vertices~$ U' $.
We then return $ U = \{ v \in V \mid v_{\mathrm{out}} \in U' \} $.

This variant of the procedure of Theorem~\ref{thm:local procedure edge-out component} is as follows.
Observe first that instead of giving a bound on the edge size of the $k$-edge-out component in Algorithm~\ref{alg:local procedure}, we could also give a bound $ \Delta' $ on the volume of the $k$-edge-out component.
This gives a component of maximum volume $ O ((k + 1) \Delta') $ and a query time and running time of $ O ((k^2 + 1) \Delta') $ for Procedure~\DetectComponent.
Furthermore, we could also give a bound $ \Delta' $ on the restricted symmetric volume.
We only need to modify the algorithm to also scan the incoming edges of every vertex as soon as it becomes an interior vertex and adding these edges to $ F_i $.
In this way, the query time and the running time still are $ O (\Delta') $.
This gives a component of maximum restricted symmetric volume $ O ((k + 1) \Delta') $ and a query time and running time of $ O ((k^2 + 1) \Delta') $.
These are $ O ((k + 1) \Delta) $ and $ O ((k^2 + 1) \Delta) $, respectively, for $ \Delta' = 3 \Delta $.

The correctness proof is as follows.
Suppose that~$ s $ is contained in a $k$-vertex-out component of (symmetric) volume at most~$ \Delta $ in $ G $.
Then by Lemma~\ref{lem:vertex component implies edge component}, $ G'_s $ has a $k$-edge-out component $ C' $ containing~$ s $ of (restricted symmetric) volume at most $ 3 \Delta = \Delta' $.
Therefore, the modified local procedure returns $ U' \supseteq \{ s_{\mathrm{out}} \} $ with probability at least~$ p $.
It follows that $ U \supseteq \{ s \} $ with probability at least $ p $ as desired.

Now suppose that $ U' \supseteq \{ s_{\mathrm{out}} \} $.
By our slight modification of Theorem~\ref{thm:local procedure edge-out component}, $ U' $ is a $k$-edge component containing $ s $ in $ G'_s $ of (restricted symmetric) volume $ O ((k + 1) \Delta') = O ((k + 1) \Delta) $.
By Lemma~\ref{lem:edge component implies vertex component}, $ U $ is a $k$-vertex-out component containing~$ s $ of (symmetric) volume $ O((k + 1) \Delta') = O ((k + 1) \Delta) $ as desired.

\section{Vertex Connectivity}\label{sec:vertex connectivity}

In the following, we give an improved algorithm for computing the vertex connectivity $ \kappa $ of a directed graph.
We obtain this algorithm by modifying the vertex connectivity algorithm of Nanongkai, Saranurak, and Yingchareonthawornchai~\cite{NanongkaiSY19}.
We review this algorithm in the following.

For a given integer $ k $, this algorithm either certifies that $ \kappa < k $ or concludes that, with high probability, $ \kappa \geq k $.
The algorithm assumes that $ k \leq \tfrac{\sqrt{m}}{2} $.
As a subroutine, the algorithm uses a local vertex cut detection algorithm, in style similar to our algorithm of Theorem~\ref{thm:local procedure vertex-out component}.
In the following, let $ \Delta^* $ be the maximum $ \Delta $ such that $ S (\Delta) + k^2 < m $, where $ S (\Delta) $ is an upper bound on the size of the component returned by the subroutine for a given $ \Delta $.
In a preprocessing step, the algorithm ensures that the graph is strongly connected (which can be checked in linear time) before proceeding.

Suppose first, the graph contains a vertex cut $ (L, M, R) $ of size at most $ k $ in which both sides $ L $ and $ R $ have symmetric volume more than $ \Delta^* $, where the symmetric volumes of $ L $ and $ R $ are defined as
\begin{equation*}
\vol^* (L) = | E (L, V) \cup E (V, L) | = | E (L, L) | + | E (L, M) | + | E (M, L) | + | E (R, L) |
\end{equation*}
and
\begin{equation*}
\vol^* (R) = | E (R, V) \cup E (V, R) | = | E (R, R) | + | E (R, M) | + | E (M, R) | + | E (R, L) | \, ,
\end{equation*}
respectively.
If we are given vertices $ s \in L $ and $ t \in R $, then we can find a vertex cut of size at most $ k $ in time $ O (k m) $ by running $ k $ iterations of the Ford-Fulkerson algorithm on a suitably modified graph.
By uniformly sampling $ O (\tfrac{m}{\Delta^*} \cdot \log n) $ pairs of edges $ (u, u') $, $ (v, v') $ and for each such sample trying out all combinations of their endpoints as choices for $ s $ and $ t $ we can find such a pair with high probability by Lemma~\ref{lem:connectivity:large cuts} below and thus certify that $ \kappa < k $.
This is the first step of the algorithm and it takes time $ O (\tfrac{m}{\Delta^*} \cdot k m \log n) $.

\begin{lemma}[Implicit in~\cite{NanongkaiSY19}]\label{lem:connectivity:large cuts}
Assume that $ G = (V, E) $ contains a vertex cut $ (L, M, R) $ of size $ |M| \leq k $ such that both $ L $ and $ R $ have symmetric volume at least $ \Delta $ and that $ k \leq \tfrac{\sqrt{m}}{2} $.
Let $ (e_1, e'_1), \ldots, (e_t, e'_t) $ for $ t = \lceil \tfrac{m}{\Delta} \cdot c \log n \rceil $ (and for any $ c \geq 1 $) be $ t $ pairs of edges that have been sampled from $ E \times E $ uniformly at random.
Then with probability at least $ 1 - \tfrac{1}{n^c} $ there is some sampled pair $ (e_i = (u_i, v_i), e'_i = (u'_i, v'_i)) $ (for $ 1 \leq i \leq t $) such that among the following four pairs of vertices there is at least one pair in which the first vertex is contained in $ L $ and the second vertex is contained in $ R $: $ (u_i, u'_i) $, $ (u_i, v'_i) $, $ (v_i, u'_i) $, $ (v_i, v'_i) $.
\end{lemma}

\begin{proof}
For ease of notation, set $ \mathcal{L} = E (L, L) \cup E (L, M) \cup E (M, L) \cup E (R, L) $ and $ \mathcal{R} = E (R, R) \cup E (R, M) \cup E (M, R) \cup E (R, L) $ in this proof.
By our assumption we know that $ | \mathcal{L} | \geq \Delta $ and $ | \mathcal{R} | \geq \Delta $.
Let $ e = (u, v) $ and $ e' = (u', v') $ be a pair of sampled edges.
We will now show that with probability at least $ \tfrac{\Delta}{4 m} $ among the following pairs of vertices there is at least one for which the first vertex is contained in $ L $ and the second vertex is contained in $ R $: $ (u, u') $, $ (u, v') $, $ (v, u') $, $ (v, v') $.
We achieve this by showing that $ e \in \mathcal{L} $ and $ e' \in \mathcal{R} $ with probability at least $ \tfrac{\Delta}{4 m} $.

Since $ e $ and $ e' $ have been sampled independently, we have $ \Pr [e \in \mathcal{L} \wedge e' \in \mathcal{R}] = \Pr [e \in \mathcal{L}] \cdot \Pr [e' \in \mathcal{R}] $.
As both $ | \mathcal{L} | \geq \Delta $ and $ | \mathcal{R} | \geq \Delta $, we have $ \Pr [e \in \mathcal{L}] \geq \tfrac{\Delta}{m} $ and $ \Pr [e' \in \mathcal{R}] \geq \tfrac{\Delta}{m} $.
We now show by a case distinction that we can get even better estimates by exploiting that if the symmetric volume on one side is relatively small, then the symmetric volume of the other side must be relatively high.

Consider first the case $ |\mathcal{L}| \geq \tfrac{m}{4} $.
Then $ \Pr [e \in \mathcal{L}] \geq \tfrac{1}{4} $ and we have 
\begin{equation*}
\Pr [e \in \mathcal{L} \wedge e' \in \mathcal{R}] = \Pr [e \in \mathcal{L}] \cdot \Pr [e' \in \mathcal{R}] \geq \frac{1}{4} \cdot \frac{\Delta}{m} = \frac{\Delta}{4m} \, .
\end{equation*}
Consider now the other case $ |\mathcal{L}| < \tfrac{m}{4} $.
Observe that the set of edges of the graph can be partitioned as follows:
\begin{equation*}
    E = \mathcal{L} \cup E (M, M) \cup E (M, R) \cup E (R, M) \cup E (R, R) \, .
\end{equation*}
As $ | E (M, M) | \leq k^2 \leq \tfrac{m}{4} $, it must be the case that $ | \mathcal{R} | \geq | E (R, R) | + | E (R, M) | + | E (M, R) | \geq | E | - | \mathcal{L} | - | E (M, M) | \geq m - \tfrac{2}{4} \cdot m \geq \tfrac{1}{4} \cdot m $.
Therefore we have
\begin{equation*}
\Pr [e \in \mathcal{L} \wedge e' \in \mathcal{R}] = \Pr [e \in \mathcal{L}] \cdot \Pr [e' \in \mathcal{R}] \geq \frac{\Delta}{m} \cdot \frac{1}{4} = \frac{\Delta}{4m}
\end{equation*}
as promised above.

Now for $ t = \lceil \tfrac{4m}{\Delta} \cdot c \ln n \rceil $ many sampled edge pairs we have that the probability that none of them fulfills the desired condition is at most
\begin{equation*}
    \left ( 1 - \frac{\Delta}{4m} \right)^{t} \leq \left ( 1 - \frac{\Delta}{4m} \right)^{\frac{4 m}{\Delta} \cdot c \ln n} \leq \frac{1}{e^{c \ln n}} = \frac{1}{n^c} \, .
\end{equation*}
Thus, the probability that for at least one sampled edge pair $ (e, e') $ we have $ e \in \mathcal{L} $ and $ e' \in \mathcal{R} $ is at least $ 1 - \tfrac{1}{n^c} $ as desired.
\end{proof}

To understand the second step of the algorithm, suppose the graph contains a vertex cut $ (L, M, R) $ of size at most $ k $ in which one of the sides $ L $ or $ R $ has symmetric volume at most $ \Delta^* $.
Then, by running the local vertex cut detection algorithm with parameter $ \Delta^* $ on some vertex of the smaller side we will find a $k$-vertex-out component or a $k$-vertex-in component of symmetric volume at most $ S (\Delta^*) < m - k^2 $.
As $ | E (M, M) | < k^2 $, this component is proper and thus certifies that $ \kappa < k $.

Let $ C \in \{ L, R \} $ be such that the symmetric volume of $ C $ is at most $ \Delta^* $ and set $ \mathcal{C} = E (C, V) \cup E (V, C) $ (where $ \vol^*(C) = | \mathcal{C} | \leq \Delta^* $).
To find a vertex contained in $ C $ efficiently, we will perform the following sampling-based method:
Set $ \Delta_i = \Delta^* / 2^i $ for every $ 0 \leq i \leq \log{\Delta^*} $.
Then $ \Delta_{i+1} \leq | \mathcal{C} | \leq \Delta_i $ for some $ i \geq 0 $.
By sampling $ O (\tfrac{m}{\Delta_{i+1}} \cdot \log n) $ many edges uniformly at random, we will with high probability find some edge contained $ \mathcal{C} $.
For every edge contained in $ \mathcal{C} $ at least one of its endpoints is contained in $ C $.
By running a local vertex cut algorithm on both endpoints of each sampled edge with parameter $ \Delta_i $, we will thus with high probability find a $k$-vertex-out component or a $k$-vertex-in component of symmetric volume at most $ S (\Delta^*) < m - k^2 $ as desired.
This is the second step of the algorithm.
If the running time for a single instance of the local vertex cut algorithm is $ T (\Delta_i) = f(k) \cdot O (\Delta_i) $ for some function $ f (\cdot) $, then the running time of the second step of the algorithm is
\begin{equation*}
O \left( \sum_{0 \leq i \leq \log{\Delta^*}} T (\Delta_i) \cdot \frac{m}{\Delta_{i+1}} \log n \right) = O \left( f(k) \cdot m \log^2 n \right)
\end{equation*}

If neither the first step, nor the second step of the algorithm have found a vertex cut of size at most $ k $, then we conclude that $ \kappa \geq k $.
Since every vertex cut of size at most $ k $ must fall into one of the two cases, it follows from the discussion above that our algorithm is correct with high probability.
Formally, the guarantees of the algorithm can be summarized as follows.

\begin{lemma}[Implicit in \cite{NanongkaiSY19}]
Suppose there is an algorithm that, given constant-time query access to a directed graph, returns, for a fixed $ k \geq 1 $, any integer~$ \Delta \geq 1 $ and any starting vertex~$ s $, in time $ T (\Delta) = f(k) \cdot O (\Delta) $ for some function $ f (k ) $ a set of vertices $ U $ such that
(1) if there is a $ k $-vertex-out component containing~$ s $ of symmetric volume at most~$ \Delta $, then $ U \supseteq \{ s \} $ with high probability and
(2) if $ U \supseteq \{ s \} $, then $ U $ is a $ k $-vertex-out component of symmetric volume at most $ S (\Delta) $.

Then there is a randomized Monte-Carlo algorithm for computing the vertex connectivity of a directed graph in time time $ O (\tfrac{m}{\Delta^*} \cdot k m \log n + f(k) \cdot m \log^2 n + n) $ where $ \Delta^* $ is the maximum $ \Delta $ such that $ S (\Delta) + k^2 < m $.
\end{lemma}

Using our local vertex cut algorithm of Theorem~\ref{thm:local procedure vertex-out component}, we have $ S (\Delta) = O (k \Delta) $ and $ T (\Delta) = O (k^2 \Delta \log n) $ to obtain a high-probability guarantee.
This means that $ \Delta^* = \Theta (\tfrac{m}{k}) $.
We therefore arrive at a running time of $ O (k^2 m \log^3 n) $.
Using a standard binary search approach, we can use this algorithm to determine $ \kappa $ (with high probability).
If in this binary search, an answer for $ k > \tfrac{\sqrt{m}}{2} $ is needed, we simply run the algorithm of Henzinger, Rao, and Gabow~\cite{HenzingerRG00} which, with high probability, computes the vertex connectivity in time $ O (m n \log{n}) $, which is $ O (k^2 m \log{n}) $ for $ k > \tfrac{\sqrt{m}}{2} $.
We thus obtain the following result.
\begin{theorem}
There is a randomized Monte-Carlo algorithm for computing the vertex connectivity~$ \kappa $ of a directed graph in time $ O (k^2 m \log^3 n \log \kappa + n) $.
The algorithm is correct with high probability.
\end{theorem}

By the standard approach of running our algorithm on the sparse $k$-vertex connectivity certificate of Nagamochi and Ibaraki~\cite{NagamochiI92}, which can be computed in linear time, we obtain the following result.

\begin{theorem}
There is a randomized Monte-Carlo algorithm for computing the vertex connectivity~$ \kappa $ of an undirected graph in time $ O ((\kappa^3 n \log^3 n + m) \log \kappa) $.
The algorithm is correct with high probability.
\end{theorem}

\section{Maximal $k$-Edge Connected Subgraphs}\label{sec:connected subgraphs}

In the following, we consider the problem of computing the maximal $k$-edge connected subgraphs of a directed and undirected graphs.
In directed graphs, we essentially follow the overall algorithmic scheme of Chechik et al.~\cite{ChechikHILP17} and obtain an improvement by plugging in our new local cut-detection procedure.
In undirected graphs, we additionally modify the algorithmic scheme to obtain further running time improvements.

\subsection{Directed Graphs}\label{sec:subgraphs directed}

The baseline recursive algorithm for computing the maximal $k$-edge connected subgraphs works as follows:
First, try to find a directed cut with at most $ k -1 $ cut edges.
If such a cut exists, remove the cut edges from the graph and recurse on each strongly connected component of the remaining graph.
If no such cut exists, then the graph is $k$-edge connected.
The recursion depth of this algorithm is at most $ n $, and using Gabow's cut algorithm~\cite{Gabow95}, it takes time $ O (k m \log n) $ to find a cut with at most $ k - 1 $ cut edges.
Therefore this algorithm has a running time of $ O (k m n \log n) $.

The idea of Chechik et al.\ is to speed up this baseline algorithm by using a local cut-detection procedure as follows:
The algorithm ensures that the graph contains no $ (k'-1) $-edge-out components of edge size at most $ \Delta $ anymore for $ k' = \min (k, \Delta) $ before Gabow's global cut algorithm is invoked.
This can be achieved as follows:
If the number of edges in the graph is at most $ S (\Delta) $ (an upper bound on the edge size of the component returned by the local cut-detection procedure), then the basic algorithm is invoked.
Otherwise, the algorithm maintains a list~$ L $ of vertices which it considers as potential starting vertices for the local cut procedure.
Initially, $ L $ consists of all vertices.
For every vertex $ s $ of $ L $ the algorithm first tries to detect a small $(k'-1)$-edge-out component containing~$ s $ and then a small $(k'-1)$-edge-in component containing~$ s $.
It then removes $ s $ from $ L $ and if a component~$ C $ was detected, it removes $ C $ from the graph (as well as the outgoing and incoming edges of~$ C $) and adds the heads of the outgoing edges and the tails of the incoming edges to~$ L $.
Each component found in this way is processed (recursively) with the baseline algorithm.
Once $ L $ is empty, Gabow's cut algorithm is invoked on the remaining graph, the cut edges are removed from the graph, their endpoints are added to a new list~$ L' $, and the strongly connected components of the remaining graph are computed.
The algorithm then recurses on each strongly connected component with the list $ L' $.
As a preprocessing step to this overall algorithm, we first compute the strongly connected components and run the algorithm on each strongly connected component separately.\footnote{We have added this preprocessing step to ensure that $ n = O (m) $ for each strongly connected component in the running time analysis.}

The running time analysis is as follows.
The strongly connected components computation in the preprocessing can be done in time $ O (m + n) $.
For every vertex, we initiate the local cut detection initially and whenever it was the endpoint of a removed edge.
We thus initiate at most $ O (n + m) = O (m) $ local cut detections, each taking time $ T (\Delta) $.
It remains to bound the time spent for the calls of Gabow's cut algorithm and the subsequent computations of strongly connected components after removing the cut edges.
On a strongly connected graph with initially $ m' $ edges, these two steps take time $ O (k m' \log n) $.
Consider all recursive calls at the same recursion level of the algorithm.
As the graphs that these recursive calls operate on are disjoint, the total time spent at this recursion level is $ O (k m \log n) $.
To bound the total recursion depth, observe that for a graph with initially $ m' $ edges, the graph passed to each recursive call has at most $ \max \{ m' - \Delta, S (\Delta) \} $ edges as the only cuts left to find for Gabow's cut algorithm either have edge size at least $ \Delta $ on one side of the cut or at least $ \Delta $ cut edges.
Thus, the recursion depth is $ O (\tfrac{m}{\Delta} + S (\Delta)) $.
Altogether, we therefore arrive at a running time of
\begin{equation*}
O \left(m \cdot T (\Delta) + \left( \frac{m}{\Delta} + S(\Delta) \right) \cdot k m \log n + n \right) \, .
\end{equation*}

Observe further that a one-sided Monte-Carlo version of the local cut-detection procedure, as the one we are giving in this paper, only affects the recursion depth.
If each execution of the procedure is successful with probability $ p \geq 1 - \tfrac{1}{n^3} $, then the probability that all $ O (m) = O (n^2) $ executions of the procedure are successful is at least $ 1 - O(\tfrac{1}{n}) $.
As the worst-case recursion depth is at most~$ n $, the expected recursion depth is at most $ O ((1 - \tfrac{1}{n}) \cdot (\tfrac{m}{\Delta} + S (\Delta)) + \tfrac{1}{n} \cdot n) = O (\tfrac{m}{\Delta} + S (\Delta)) $.

The analysis of this algorithmic scheme can be summarized in the following lemma.
\begin{lemma}[Implicit in \cite{ChechikHILP17}]
Suppose there is an algorithm that, given constant-time query access to a directed graph, returns, for a fixed $ k \geq 1 $, any integer~$ \Delta \geq k $ and any starting vertex~$ s $, in time $ T (\Delta) $ a set of vertices $ U $ such that
(1) if there is a $ (k-1) $-edge-out component containing~$ s $ of edge size at most~$ \Delta $, then $ U \supseteq \{ s \} $ with probability at least $ 1 - \tfrac{1}{n^3} $ and
(2) if $ U \supseteq \{ s \} $, then $ U $ is a $ (k-1) $-edge-out component of edge size at most $ S (\Delta) $.

Then there is an algorithm for computing the maximal $k$-edge connected subgraphs of a directed graph with expected running time $ O (m \cdot T (\Delta) + (\tfrac{m}{\Delta} + S(\Delta)) \cdot k m \log n + n) $ for every $ 1 \leq \Delta \leq m $.
\end{lemma}

By plugging in the local cut-detection procedure of Chechik et al.\ mentioned in Lemma~\ref{lem:local procedure edge-out component Chechik} with $ T (\Delta) = O ((2k)^{k+1} \Delta) $ and $ S (\Delta) = O (k \Delta) $, one obtains an algorithm with running time $ O ((2k)^{k + 2} m^{3/2} \log{n}) $.

For any $ k \geq \Delta $, our improved local cut-detection procedure of Theorem~\ref{thm:local procedure edge-out component} has $ T (\Delta) = O (k^2 \Delta \log n) $ and $ S (\Delta) = O (k \Delta) $.
By setting $ \Delta = \tfrac{\sqrt{m}}{\sqrt{k}} $, we arrive at running time of
\begin{equation*}
O \left( k^2 m \Delta \log n + \left( \frac{m}{\Delta} + k \Delta \right) \cdot k m \log n + n \right) = O (k^{3/2} m^{3/2} \log n + n) \, .
\end{equation*}

\begin{theorem}
There is a randomized Las Vegas algorithm for computing the maximal $k$-edge connected subgraphs of a directed graph that has expected running time $ O (k^{3/2} m^{3/2} \log n + n) $.
\end{theorem}

\subsection{Undirected Graphs}

In undirected graphs, we obtain a tighter upper bound on the running time of the algorithm for three reasons.
First, instead of parameterizing the algorithm by and edge size $ \Delta $, we parameterize it by a vertex size $ \Gamma $.
Thus, whenever the list $ L $ in the algorithm becomes empty we can be sure that there is no component of vertex size at most $ \Gamma $ in the current graph anymore.
Components that are larger can be found at most $ \tfrac{n}{\Gamma} $ times.
Second, each $(k-1)$-edge-out component detected by the local cut procedure is also a $(k-1)$-edge-in component, i.e., its number of incoming edges equals its number of outgoing edges and is at most $ k - 1 $.
Thus, the number of removed edges per successful component detection is at most $ k - 1 $.
As there are at most $ n $ such successful detections in total, the total number of executions of the local cut-detection procedure is at most $ n + (k-1) n = k n $.
Third, in undirected graphs we can run instances of both Gabow's global cut detection algorithm and our local cut detection algorithm on a sparse $k$-edge connectivity certificate~\cite{Thurimella89,NagamochiI92} of the current graph.
The sparse certificate can be computed in time $ O (m + n) $, but we do not want to explicitly perform this expensive computation each time edges are removed from the graph. 
Instead, we maintain the sparse certificate with a dynamic algorithm as outlined below.\footnote{Let us emphasize again that the sparse certificate is not w.r.t.\ to the input graph but w.r.t.\ to the graph in which all cut edges found by Gabow's algorithm and all edges incident on components detected by the local algorithm have been removed.}
As the $k$-edge certificate has size $ O (k n) $, Gabow's cut algorithm has running time $ O (k^2 n \log{n}) $ if it is run on the certificate.
Furthermore, as observed by Nanongkai et al.~\cite{NanongkaiSY19}, the $k$-edge certificate has arboricity $ k $.
Since the arboricity bounds the local density of any vertex-induced subgraph, the edge size of a component of vertex size $ \Gamma $ is $ O (\Gamma k) $ on the $k$-edge certificate.
Each instance of the local cut detection procedure therefore has running time $ T (k \Gamma) $ and if successful detects a component of size $ S (k \Gamma) $.

A sparse $k$-connectivity certificate of a graph $ G = (V, E) $ is a graph $ H = (V, F_1 \cup \ldots \cup F_k) $ such that for every $ 1 \leq i \leq k $ the graph $ (V, F_i) $ is a spanning forest of $ (V, E \setminus \bigcup_{1 \leq j \leq i- 1} F_j) $ (where in particular $ (V, F_1) $ is a spanning forest of $ G $).
The dynamic connectivity algorithm of Holm et al.~\cite{HolmLT01} can be used to dynamically maintain a spanning forest of a graph undergoing edge insertions and deletions in time $ O (\log^2 n) $ per update.
Maintaining the hierarchy of spanning forests for a $k$-edge certificate under a sequence edge removals to $ G $ takes time $ O (k m \log^2 {n}) $ by the following argument\footnote{This argument is similar to the one for maintaining the cut sparsifier in~\cite{AbrahamDKKP16}.}:
The dynamic algorithm of Holm et al.\ makes at most one change to the spanning forest per change to the input graph.
Therefore each deletion in $ G $ causes at most one update to to each of the $ k $ levels of the hierarchy.
As at most $ m $ edges can be removed from~$ G $, the total update time is $ O (k m \log^2 {n}) $.

Using otherwise the same analysis as in Section~\ref{sec:subgraphs directed}, we arrive at the following running time:
\begin{equation*}
O \left(k n \cdot T (k \Gamma) + \left( \frac{n}{\Gamma} + S(k \Gamma) \right) \cdot k^2 n \log n + k m \log^2 n \right) \, .
\end{equation*}
The analysis of this algorithmic scheme can be summarized in the following lemma.
\begin{lemma}
Suppose there is an algorithm that, given constant-time query access to an undirected graph, returns, for a fixed $ k \geq 1 $, any integer~$ \Delta \geq k $ and any starting vertex~$ s $, in time $ T (\Delta) $ a set of vertices $ U $ such that
(1) if there is a $ (k-1) $-edge-out component containing~$ s $ of edge size at most~$ \Delta $, then $ U \supseteq \{ s \} $ with probability at least $ 1 - \tfrac{1}{n^3} $ and
(2) if $ U \supseteq \{ s \} $, then $ U $ is a $ (k-1) $-edge-out component of edge size at most $ S (\Delta) $.

Then there is a randomized Las Vegas algorithm for computing the maximal $k$-edge connected subgraphs of an undirected graph with expected running time $ O (k n \cdot T (k \Gamma) + (\tfrac{n}{\Gamma} + S(k \Gamma)) \cdot k^2 n \log n + k m \log^2 n) $ for every $ 1 \leq \Gamma \leq n $.
\end{lemma}

For any $ \Gamma \geq 1 $, our local cut-detection procedure of Theorem~\ref{thm:local procedure edge-out component} has $ T (k \Gamma) = O (k^3 \Gamma \log n) $ and $ S (k \Gamma) = O (k^2 \Gamma) $.
Thus, the running time is
\begin{equation*}
O \left(k^4 n \Gamma \log n + \left( \frac{n}{\Gamma} + k^2 \Gamma \right) \cdot k^2 n \log n + k m \log^2 n \right) \, .
\end{equation*}
By setting $ \Gamma = \tfrac{\sqrt{n}}{k} $, we obtain a running time of $ O (k^4 n^{3/2} \log{n} + k m \log^2 n) $.

\begin{theorem}
There is a randomized Las Vegas algorithm for computing the maximal $k$-edge connected subgraphs of an undirected graph that has expected running time $ O (k^4 n^{3/2} \log{n} + k m \log^2 n) $.
\end{theorem}

\paragraph{Remark on Sparse Certificates}

Our algorithm tailored to undirected graphs relies on sparse certificates.
It would be tempting to run our algorithm for directed graphs directly on a sparse certificate to achieve running-time savings by performing a maximum amount of sparsification.
However this approach does not work as Figure~\ref{fig:sparse certificate} shows.
Sparse certificates preserve the $k$-edge connectivity, but not necessarily the maximal $k$-edge connected subgraphs.

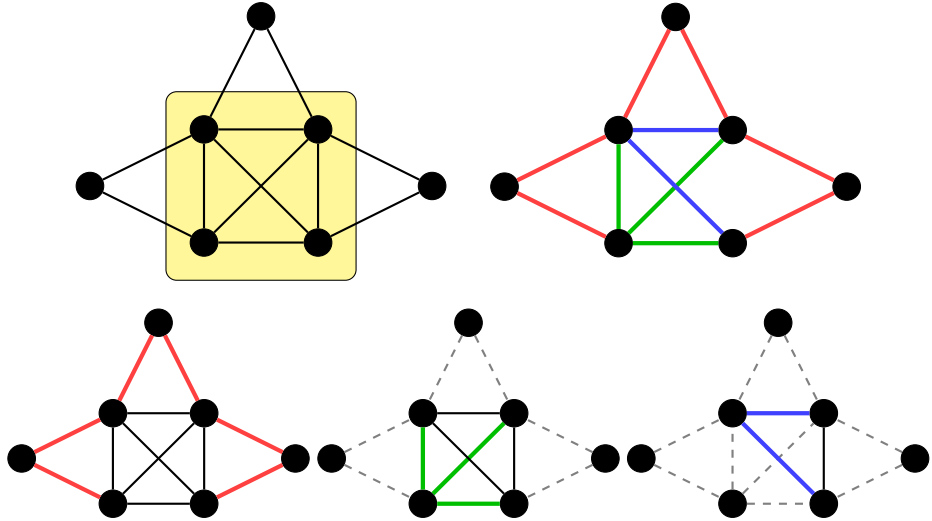
\begin{figure}[htbp!]
    \centering
    \begin{tikzpicture}
	\tikzset{default_node/.style={circle,draw,fill=black,minimum size=1pt}}
	\tikzset{edge/.style={thick,black}}
	
	\draw[rounded corners,fill=yellow!50] (-0.5,-2) rectangle (2,0.5);
	
	\node[default_node] (1) at (0,-1.5) {};
	\node[default_node] (2) at (0,0) {};
	\node[default_node] (3) at (1.5,0) {};
	\node[default_node] (4) at (1.5,-1.5) {};
	
	\node[default_node] (5) at (-1.5,-0.75) {};
	\node[default_node] (6) at (0.75,1.5) {};
	\node[default_node] (7) at (3.0,-0.75) {};

	\draw[edge] (1) edge (2);
	\draw[edge] (1) edge (3);
	\draw[edge] (1) edge (4);
	\draw[edge] (2) edge (3);
	\draw[edge] (2) edge (4);
	\draw[edge] (3) edge (4);
	
	\draw[edge] (5) edge (1);
	\draw[edge] (5) edge (2);
	\draw[edge] (6) edge (2);
	\draw[edge] (6) edge (3);
	\draw[edge] (7) edge (3);
	\draw[edge] (7) edge (4);
    \end{tikzpicture}
    \hspace{1em}
    \begin{tikzpicture}
	\tikzset{default_node/.style={circle,draw,fill=black,minimum size=1pt}}
	\tikzset{edge/.style={thick,black}}
	\tikzset{spanning_tree_edge1/.style={ultra thick,red!75}}
	\tikzset{spanning_tree_edge2/.style={ultra thick,green!75!black}}
	\tikzset{spanning_tree_edge3/.style={ultra thick,blue!75}}

	\draw[rounded corners,draw=none] (-0.5,-2) rectangle (2,0.5);

	\node[default_node] (1) at (0,-1.5) {};
	\node[default_node] (2) at (0,0) {};
	\node[default_node] (3) at (1.5,0) {};
	\node[default_node] (4) at (1.5,-1.5) {};
	
	\node[default_node] (5) at (-1.5,-0.75) {};
	\node[default_node] (6) at (0.75,1.5) {};
	\node[default_node] (7) at (3.0,-0.75) {};

	\draw[spanning_tree_edge2] (1) edge (2);
	\draw[spanning_tree_edge2] (1) edge (3);
	\draw[spanning_tree_edge2] (1) edge (4);
	\draw[spanning_tree_edge3] (2) edge (3);
	\draw[spanning_tree_edge3] (2) edge (4);
	
	\draw[spanning_tree_edge1] (5) edge (1);
	\draw[spanning_tree_edge1] (5) edge (2);
	\draw[spanning_tree_edge1] (6) edge (2);
	\draw[spanning_tree_edge1] (6) edge (3);
	\draw[spanning_tree_edge1] (7) edge (3);
	\draw[spanning_tree_edge1] (7) edge (4);
    \end{tikzpicture}
    
    \vspace{2ex}
    
    \begin{tikzpicture}[scale=0.8]
	\tikzset{default_node/.style={circle,draw,fill=black,minimum size=1pt}}
	\tikzset{edge/.style={thick,black}}
	\tikzset{spanning_tree_edge1/.style={ultra thick,red!75}}
	\node[default_node] (1) at (0,-1.5) {};
	\node[default_node] (2) at (0,0) {};
	\node[default_node] (3) at (1.5,0) {};
	\node[default_node] (4) at (1.5,-1.5) {};
	
	\node[default_node] (5) at (-1.5,-0.75) {};
	\node[default_node] (6) at (0.75,1.5) {};
	\node[default_node] (7) at (3.0,-0.75) {};

	\draw[edge] (1) edge (2);
	\draw[edge] (1) edge (3);
	\draw[edge] (1) edge (4);
	\draw[edge] (2) edge (3);
	\draw[edge] (2) edge (4);
	\draw[edge] (3) edge (4);
	
	\draw[spanning_tree_edge1] (5) edge (1);
	\draw[spanning_tree_edge1] (5) edge (2);
	\draw[spanning_tree_edge1] (6) edge (2);
	\draw[spanning_tree_edge1] (6) edge (3);
	\draw[spanning_tree_edge1] (7) edge (3);
	\draw[spanning_tree_edge1] (7) edge (4);
    \end{tikzpicture}
    \begin{tikzpicture}[scale=0.8]
	\tikzset{default_node/.style={circle,draw,fill=black,minimum size=1pt}}
	\tikzset{edge/.style={thick,black}}
	\tikzset{spanning_tree_edge1/.style={thick,dashed,black!50}}
	\tikzset{spanning_tree_edge2/.style={ultra thick,green!75!black}}
	\node[default_node] (1) at (0,-1.5) {};
	\node[default_node] (2) at (0,0) {};
	\node[default_node] (3) at (1.5,0) {};
	\node[default_node] (4) at (1.5,-1.5) {};
	
	\node[default_node] (5) at (-1.5,-0.75) {};
	\node[default_node] (6) at (0.75,1.5) {};
	\node[default_node] (7) at (3.0,-0.75) {};

	\draw[spanning_tree_edge2] (1) edge (2);
	\draw[spanning_tree_edge2] (1) edge (3);
	\draw[spanning_tree_edge2] (1) edge (4);
	\draw[edge] (2) edge (3);
	\draw[edge] (2) edge (4);
	\draw[edge] (3) edge (4);
	
	\draw[spanning_tree_edge1] (5) edge (1);
	\draw[spanning_tree_edge1] (5) edge (2);
	\draw[spanning_tree_edge1] (6) edge (2);
	\draw[spanning_tree_edge1] (6) edge (3);
	\draw[spanning_tree_edge1] (7) edge (3);
	\draw[spanning_tree_edge1] (7) edge (4);
    \end{tikzpicture}
    \begin{tikzpicture}[scale=0.8]
	\tikzset{default_node/.style={circle,draw,fill=black,minimum size=1pt}}
	\tikzset{edge/.style={thick,black}}
	\tikzset{spanning_tree_edge1/.style={thick,dashed,black!50}}
	\tikzset{spanning_tree_edge2/.style={thick,dashed,black!50}}
	\tikzset{spanning_tree_edge3/.style={ultra thick,blue!75}}

	\node[default_node] (1) at (0,-1.5) {};
	\node[default_node] (2) at (0,0) {};
	\node[default_node] (3) at (1.5,0) {};
	\node[default_node] (4) at (1.5,-1.5) {};
	
	\node[default_node] (5) at (-1.5,-0.75) {};
	\node[default_node] (6) at (0.75,1.5) {};
	\node[default_node] (7) at (3.0,-0.75) {};

	\draw[spanning_tree_edge2] (1) edge (2);
	\draw[spanning_tree_edge2] (1) edge (3);
	\draw[spanning_tree_edge2] (1) edge (4);
	\draw[spanning_tree_edge3] (2) edge (3);
	\draw[spanning_tree_edge3] (2) edge (4);
	\draw[edge] (3) edge (4);
	
	\draw[spanning_tree_edge1] (5) edge (1);
	\draw[spanning_tree_edge1] (5) edge (2);
	\draw[spanning_tree_edge1] (6) edge (2);
	\draw[spanning_tree_edge1] (6) edge (3);
	\draw[spanning_tree_edge1] (7) edge (3);
	\draw[spanning_tree_edge1] (7) edge (4);
    \end{tikzpicture}

    \caption{The top row shows an example of an undirected graph on the left with maximal $3$-edge component consisting of four vertices (marked by the yellow box) that is not preserved in its sparse certificate on the right. The bottom row shows the $ 3 $~spanning forests found for the sparse certificate.}
    \label{fig:sparse certificate}
\end{figure}

\section{Testing Higher Connectivity}\label{sec:property testing}

We now explain how our new approach for locally detecting edge-out and vertex-out components leads to improved property testing algorithms for $k$-edge connectivity and $k$-vertex connectivity.

\subsection{Testing $k$-Edge Connectivity}

We build upon the property testing algorithm of Orenstein and Ron~\cite{OrensteinR11}, which we will briefly review in the following.
Their tester exploits the following structural property.

\begin{lemma}[\cite{OrensteinR11}]\label{lem:number of components}
A directed graph that is $ \epsilon $-far from being $k$-edge connected has at least $ \tfrac{\epsilon m}{2 k} $ subsets of vertices that are proper $(k-1)$-edge-out components or proper $(k-1)$-edge-in components of vertex size at most $ \frac{2 k}{\epsilon \davg} $.
\end{lemma}

The main idea now is that if the input graph is $\epsilon$-far from being $k$-edge connected, random sampling of a sufficient number of vertices will ``hit'' one of these small components with constant probability.
By running a local search procedure, the algorithm can detect for each sampled vertex whether it is contained in a small component.
The algorithm thus has the following steps:
\begin{enumerate}
    \item Sample $ \Theta (\tfrac{k}{\epsilon \davg}) $ vertices uniformly at random.
    \item For each sampled vertex $ s $, check if $ s $ is contained in a proper $ (k-1) $-edge-out component or a in a proper $ (k-1) $-edge-out component of vertex size at most~$ \tfrac{2 k}{\epsilon \davg}$.
    \item If this is the case for any sampled vertex, then \textsc{Reject} and \textsc{Accept} otherwise.
\end{enumerate}
To implement Step~2 of their algorithm, Orenstein and Ron design a procedure that performs $ O ( ( \tfrac{c k}{\epsilon \davg} )^{k+1} ) $ queries to the graph (for some constant $ c $).
As this procedure is run for each sampled vertex, the overall number of queries of the property testing algorithm is $ O ( ( \tfrac{c k}{\epsilon \davg} )^{k+2} ) $.
Using a standard technique, Orenstein and Ron reduce this to $ O ( ( \tfrac{c k}{\epsilon \davg} )^{k+1} \log (\frac{k}{\epsilon \davg}) ) $.

We observe that the requirements of the tester in Step~2 can be relaxed as follows: If a sampled vertex $ s $ is contained in \emph{any} proper $ (k-1) $-edge-out component or $ (k-1) $-edge-in component, then the tester may reject as the graph is not $k$-edge connected if such a component exists; it is not necessary for this component to have vertex size at most~$ \tfrac{2 k}{\epsilon \davg} $.
The modified algorithm works as follows:
\begin{enumerate}
    \item Sample $ \Theta (\tfrac{k}{\epsilon \davg}) $ vertices uniformly at random.
    \item For each sampled vertex $ s $, run a decision procedure that has the following guarantees:
    \begin{enumerate}
        \item If $ s $ is contained in a proper $ (k-1) $-edge-out component of vertex size at most~$ \tfrac{2 k}{\epsilon \davg} $ or a in a $ (k-1) $-edge-in component of vertex size at most~$ \tfrac{2 k}{\epsilon \davg} $, then the procedure returns \textsc{Yes}.\label{itm:first case}
        \item If the procedure returns \textsc{Yes}, then $ s $ is contained in a proper $ (k-1) $-edge-out component (of arbitrary size).
    \end{enumerate}
    \item If the decision procedure returns \textsc{Yes} for any sampled vertex, then \textsc{Reject} and \textsc{Accept} otherwise.
\end{enumerate}

It is sufficient to have a local procedure for $(k-1)$-edge-out components as the corresponding algorithm for $(k-1)$-edge in components can be obtained by running the former procedure on the reverse graph.\footnote{Recall that in the property testing model considered in this paper both outgoing and incoming edges can be queried.}
Furthermore, to obtain a tester with false-reject probability at most $ \tfrac{1}{3} $ it suffices to correctly identify a $ (k-1) $-edge-out component or a $ (k-1) $-edge-in component of vertex size at most~$ \tfrac{2 k}{\epsilon \davg} $ containing a given starting vertex in case~(a) with probability at least~$ \tfrac{5}{6} $, similar to the tester of Goldreich and Ron~\cite{GoldreichR02} for undirected graphs.
The correctness of the proposed algorithm is immediate and its guarantees can be summarized as follows.
\begin{lemma}\label{lem:property testing with black box}
Suppose we are given a decision procedure that, given a starting vertex~$ s $ and query access to a directed graph~$ G $, performs at most $ Q $ queries to~$ G $ and has the following guarantees:
\begin{enumerate}
    \item[(a)] If $ s $ is contained in a proper $ (k-1) $-edge-out component of vertex size at most~$ \tfrac{2 k}{\epsilon \davg} $, then the procedure returns \textsc{Yes} with probability at least~$ \tfrac{5}{6} $.
    \item[(b)] If the procedure returns \textsc{Yes}, then $ s $ is contained in a proper $ (k-1) $-edge-out component (of arbitrary size).
\end{enumerate}
Then there is a one-sided property testing algorithm for $k$-edge connectivity with false-reject probability at most $ \tfrac{1}{3} $ performing $ O (\frac{k}{\epsilon \davg} \cdot Q) $ many queries.
\end{lemma}

Thus, the question of designing an efficient tester, can be reduced to finding an efficient local search procedure with properties (a) and (b).
Such a procedure readily follows from our local algorithm of Theorem~\ref{thm:local procedure edge-out component}; to plug this algorithm into Lemma~\ref{lem:property testing with black box}, observe that any component of vertex size at most $ \tfrac{2 k}{\epsilon \davg} $ has edge size at most $ (\tfrac{2 k}{\epsilon \davg})^2 $.
For $ \Delta := (\tfrac{2 k}{\epsilon \davg})^2 $, the algorithm performs $ O (k^2 (\Delta + k)) = O (k^2 ((\tfrac{k}{\epsilon \davg})^2 + k )) $ queries and returns a set of vertices of edge size at most $ 2k (\Delta + k) = O (k ((\tfrac{k}{\epsilon \davg})^2 + k )) $.
We may assume without loss of generality that $ k > \frac{\epsilon \davg}{2} $ as otherwise any $ \epsilon $-far graph contains $ n $ singleton components by Lemma~\ref{lem:number of components} and we can thus reject such a graph in a preprocessing step by querying any arbitrary vertex and checking if its degree is less than $ k $.
We therefore have $ \Delta \geq k $ and thus the bound on the number of queries is $ O (k^2 \Delta) = O (\tfrac{k^4}{(\epsilon \davg)^2}) $ and the bound on the edge size is $ O (k \Delta) = O (\tfrac{k^3}{(\epsilon \davg)^2}) $.

We now do the following:
If $ m \leq 2k (\Delta + k) $, then we query the whole graph (performing $ 2k (\Delta + k) $ queries) and directly check if $ s $ is contained in a proper $ (k-1) $-edge-out component of edge size at most~$ \Delta $ using any centralized algorithm.\footnote{Recall that in the property testing model considered in this paper the tester knows $ m $, the number of edges.}
Otherwise, we run the algorithm of Theorem~\ref{thm:local procedure edge-out component} with $ \Delta := (\tfrac{2 k}{\epsilon \davg})^2 $ and return \textsc{Yes} whenever it finds a set of vertices containing~$ s $ and \textsc{No} otherwise.
In this way, properties (a) and (b) are satisfied.
As this is repeated $ O (\tfrac{k}{\epsilon \davg}) $ times, we obtain a testing algorithm with the following guarantees.

\begin{theorem}
There is a one-sided property testing algorithm for $k$-edge connectivity with false-reject probability at most $ \tfrac{1}{3} $ that performs $ O (\tfrac{k^5}{(\epsilon \davg)^3}) $ many queries.
\end{theorem}

Similar to Orenstein and Ron~\cite{OrensteinR11}, we can (1) employ a technique of Goldreich and Ron~\cite{GoldreichR02} to reduce the query time by a doubling approach for the size of the components to detect and (2) employ a technique of Yoshida and Ito~\cite{YoshidaI10} to run the tester on bounded-degree graphs by setting $ \epsilon' = \epsilon / 13 $.
The guarantees obtained this way can be summarized as follows.
\begin{lemma}\label{lem:edge connectivity testing with black box}
Suppose we are given a decision procedure parameterized by $ \Gamma \geq 1 $ that, given a starting vertex~$ s $ and query access to a directed graph~$ G $, performs at most $ Q (\Gamma) $ queries to $ G $ and has the following guarantees:
\begin{enumerate}
    \item[(a)] If $ s $ is contained in a proper $ (k-1) $-edge-out component of vertex size at most~$ \Gamma $ in~$ G $, then the procedure returns \textsc{Yes} with probability at least~$ \tfrac{5}{6} $.
    \item[(b)] If the procedure returns \textsc{Yes}, then $ s $ is contained in a proper $ (k-1) $-edge-out component in~$ G $ (of arbitrary size).
\end{enumerate}
Then there is a one-sided property testing algorithm for $k$-edge connectivity with false-reject probability at most $ \tfrac{1}{3} $ performing
\begin{equation*}
\sum_{i=1}^{O( \log{(k / (\epsilon \davg)))}} Q (2^i - 1) \cdot O \left( \frac{ k \log{(k / (\epsilon \davg))}}{2^i \epsilon \davg} \right)
\end{equation*}
many queries in unbounded-degree graphs and
\begin{equation*}
\sum_{i=1}^{O( \log{(k / (\epsilon d)))}} Q (2^i - 1) \cdot O \left( \frac{ k \log{(k / (\epsilon d))}}{2^i \epsilon d} \right)
\end{equation*}
many queries in bounded-degree graphs.
\end{lemma}

Since a $ (k - 1) $-edge-out component of vertex size $ \Gamma $ has edge size at most $ \Gamma^2 $ in general and at most $ \Gamma d $ in bounded-degree graphs (where we may additionally assume $ d \geq k $), we arrive at the following result.

\begin{theorem}\label{thm:main result testing edge connectivity}
There is a one-sided property testing algorithm for $k$-edge connectivity with false-reject probability at most $ \tfrac{1}{3} $ that performs $ O ( \tfrac{k^4}{(\epsilon \davg)^2} \cdot \log (\frac{k}{\epsilon \davg})) $ many queries in unbounded-degree graphs and $ O (\tfrac{k^3}{\epsilon} \cdot (\log{(\tfrac{k}{\epsilon d})})^2) $ many queries in bounded-degree graphs.
\end{theorem}

By a simple reduction, this result can be extended to undirected graphs: replace every undirected edge $ e = \{ u, v \} $ by two antiparallel edges $ (u, v) $ and $ (v, u) $.
For $k$-vertex connectivity, Orenstein and Ron~\cite{OrensteinR11} have argued that this correctly reduces the undirected case to the directed case.
Literally the same arguments also work for $k$-edge connectivity.

\subsection{Testing $k$-Vertex Connectivity}

The property testing algorithm for $k$-vertex connectivity follows the same principles as the algorithm for $k$-edge connectivity.
It exploits the following combinatorial property.

\begin{lemma}[\cite{OrensteinR11}]
A directed graph that is $ \epsilon $-far from being $k$-vertex connected has at least $ \tfrac{\epsilon m}{2 k} $ subsets of vertices that are $(k-1)$-vertex out components or $(k-1)$-vertex in components of vertex size at most $ \frac{2 k}{\epsilon \davg} $.
\end{lemma}

By carrying out the analysis in the same way as Orenstein and Ron-- together with the additional observation that the requirements on the local procedure for detecting a $ (k-1) $-vertex component can be relaxed -- we obtain the following property testing algorithm.\footnote{Note however that for technical reasons the constants hidden in the $ O (\cdot) $-notation differ from those of Lemma~\ref{lem:edge connectivity testing with black box} in bounded-degree graphs.} 

\begin{lemma}\label{lem:generic tester vertex connectivity}
Suppose we are given a decision procedure parameterized by $ \Gamma \geq 1 $ that, given a starting vertex~$ s $ and query access to a directed graph~$ G $, performs at most $ Q (\Gamma) $ queries to $ G $ and has the following guarantees:
\begin{enumerate}
    \item[(a)] If $ s $ is contained in a proper $ (k-1) $-vertex-out component of vertex size at most~$ \Gamma $ in~$ G $, then the procedure returns \textsc{Yes} with probability at least~$ \tfrac{5}{6} $.
    \item[(b)] If the procedure returns \textsc{Yes}, then $ s $ is contained in a proper $ (k-1) $-edge-out component in~$ G $ (of arbitrary size).
\end{enumerate}
Then there is a one-sided property testing algorithm for $k$-vertex connectivity with false-reject probability at most $ \tfrac{1}{3} $ performing
\begin{equation*}
\sum_{i=1}^{O ( \log{(k / (\epsilon \davg))} )} Q (2^i - 1) \cdot O \left( \frac{k \log{(k / (\epsilon \davg))}}{2^i \epsilon \davg} \right)
\end{equation*}
many queries in unbounded-degree graphs and
\begin{equation*}
\sum_{i=1}^{O ( \log{(k / (\epsilon d))} )} Q (2^i - 1) \cdot O \left( \frac{k \log{(k / (\epsilon d))}}{2^i \epsilon d} \right)
\end{equation*}
many queries in bounded-degree graphs.
\end{lemma}

Observe that a $ (k - 1) $-vertex-out component of vertex size $ \Gamma $ has volume at most $ \Gamma^2 + \Gamma k \leq 2 \Gamma^2 k $ in general and at most $ \Gamma d $ in bounded-degree graphs. 
By plugging in the local procedure of Theorem~\ref{thm:local procedure vertex-out component} we arrive at the following result.

\begin{theorem}\label{thm:main result testing vertex connectivity}
There is a one-sided property testing algorithm for $k$-vertex connectivity with false-reject probability at most $ \tfrac{1}{3} $ that performs $ O ( \tfrac{k^5}{(\epsilon \davg)^2} \cdot \log (\frac{k}{\epsilon \davg})) $ many queries in unbounded-degree graphs and $ O (\tfrac{k^3}{\epsilon} \cdot (\log{(\tfrac{k}{\epsilon d})})^2) $ many queries in bounded-degree graphs.
\end{theorem}

By a simple reduction of Orenstein and Ron~\cite{OrensteinR11}, this result can be extended to undirected graphs.

\section{Conclusion}\label{sec:conclusion}

In this paper, we have given a faster local procedure for detecting cuts of bounded size.
Besides being a natural object of study in its own right, we have demonstrated that this problem is pivotal to several algorithms for higher-connectivity problems that use such a procedure as an essential subroutine.
Improving its running-time guarantees has immediate consequences on several important problems.
In addition to obtaining new results for these problems, we believe that our main contribution is a synthesis and simplification of prior work.

Based on the significance of the local cut detection problem, we suggest the following research directions.
First, our procedure has a slack of $ O (k) $, i.e., it might return a component of edge size up to $ O (k) \cdot \Delta $ instead of tightly respecting the upper bound of $ \Delta $.
This seems somewhat natural, but we would like to know if there is a deeper reason for this.
Is there a matching lower bound on the query complexity?
Second, even when allowing a slack of $ O (k) $, our procedure runs in time $ O (k^2 \Delta) $, i.e., it has an overhead of another factor of $ k $.
In contrast, the local vertex cut procedure of Nanongkai et al.~\cite{NanongkaiSY19}, which follows a different approach, runs in time $ O (k \Delta^{3/2}) $, i.e., the additional overhead is $ \sqrt{\Delta} $ instead of $ k $.
Can we get the best of both worlds with a running time of $ O (k \Delta) $?
Third, randomization is used in our local procedure in an elegant, but relatively simple way.
However, we do not see any straightforward derandomization approach.
Can the local cut detection be performed deterministically while still keeping linear dependence on $ \Delta $ in the running time and without paying the exponential blowup in the dependence on $ k $ as in the argument of Chechik et al.~\cite{ChechikHILP17}?

\subsection*{Acknowledgments}

We would like to thank Asaf Ferber for fruitful discussions.
The first author would like to thank Gramoz Goranci for giving comments on a previous draft of this manuscript.

\bibliography{references}
\bibliographystyle{alpha}

\end{document}